\documentclass[12pt, twoside]{article}
\usepackage{amsmath,amsthm,amssymb,bbm}
\usepackage{times}
\usepackage{enumerate}

\def\titlerunning#1{\gdef\titrun{#1}}
\makeatletter
\def\author#1{\gdef\autrun{\def\and{\unskip, }#1}\gdef\@author{#1}}
\def\address#1{{\def\and{\\\hspace*{18pt}}\renewcommand{\thefootnote}{}%
\footnote {#1}}%
\markboth{\autrun}{\titrun}}
\makeatother
\def\email#1{e-mail: #1}
\def\subjclass#1{{\renewcommand{\thefootnote}{}%
\footnote{\emph{Mathematics Subject Classification (2010):} #1}}}
\def\keywords#1{\par\medskip
\noindent\textbf{Keywords.} #1}

\newtheorem{theorem}{Theorem}[section]
\newtheorem{proposition}[theorem]{Proposition}
\newtheorem{corollary}[theorem]{Corollary}
\newtheorem{lemma}[theorem]{Lemma}

\theoremstyle{definition}
\newtheorem{definition}[theorem]{Definition}
\newtheorem{remark}[theorem]{Remark}
\newtheorem{example}[theorem]{Example}

\numberwithin{equation}{section}

\usepackage[linktocpage=true,colorlinks=true, linkcolor=blue, citecolor=red, urlcolor=green]{hyperref}
\usepackage[all]{xy}

\renewcommand{\ker}{\Ker}

\newcommand{\mc}[1]{\mathcal{#1}}
\newcommand{\mf}[1]{\mathfrak{#1}}
\newcommand{\mb}[1]{\mathbb{#1}}
\newcommand{\id}{\mathbbm{1}}

\newcommand{\End}{\mathop{\rm End}}
\newcommand{\ad}{\mathop{\rm ad}}
\newcommand{\Ad}{\mathop{\rm Ad}}
\newcommand{\im}{\mathop{\rm Im}}
\newcommand{\Vect}{\mathop{\rm Vect}}
\newcommand{\Ker}{\mathop{\rm Ker}}
\newcommand{\rk}{\mathop{\rm rk}}
\newcommand{\Zhu}{\mathop{\rm Zhu}}

\newcommand{\Span}{\mathop{\rm Span }}

\begin{document}


\baselineskip=17pt


\titlerunning{Classical $\mc W$-algebras}

\title{Structure of classical (finite and affine) $\mc W$-algebras}

\author{
Alberto De Sole
\and
Victor G. Kac
\and 
Daniele Valeri
}

\date{}

\maketitle

\address{
A. De Sole:
Dipartimento di Matematica, Sapienza Universit\`a di Roma,
P.le Aldo Moro 2, 00185 Rome, Italy
and
Department of Mathematics, MIT,
77 Massachusetts Avenue, Cambridge, MA 02139, USA;
\email{
desole@mat.uniroma1.it
}
\and
V. Kac:
Department of Mathematics, MIT,
77 Massachusetts Avenue, Cambridge, MA 02139, USA;
\email{
kac@math.mit.edu
}
\and
D. Valeri:
SISSA, Via Bonomea 265, 34136 Trieste, Italy;
\email{
dvaleri@sissa.it
}
}



\begin{abstract}
First, we derive an explicit formula for the Poisson bracket of the 
classical finite $\mc W$-algebra $\mc W^{\text{fin}}(\mf g,f)$,
the algebra of polynomial functions on the Slodowy slice
associated to a simple Lie algebra $\mf g$ and its nilpotent element $f$. 
On the other hand, we produce an explicit set of generators
and we derive an explicit formula for the Poisson vertex algebra structure
of the classical affine $\mc W$-algebra $\mc W(\mf g,f)$.
As an immediate consequence, we obtain a Poisson algebra isomorphism
between $\mc W^{\text{fin}}(\mf g,f)$ and the Zhu algebra of $\mc W(\mf g,f)$.
We also study the generalized Miura map for classical $\mc W$-algebras.

\end{abstract}


\section{Introduction}\label{sec:1}

The four fundamental frameworks of physical theories,
classical mechanics, classical field theory,
quantum mechanics, and quantum field theory,
have, as their algebraic counterparts,
respectively, the following four fundamental algebraic structures:
Poisson algebras (PA),
Poisson vertex algebras (PVA), 
associative algebras (AA),
and vertex algebras (VA).
We thus have the following diagram:
\begin{equation}\label{maxi}
\UseTips
\xymatrix{
PVA\,\,\,\,\,\, 
\ar[d]_{\text{Zhu}}
&
\ar[l]_{\text{cl.limit}}
\,\,\,\,\,\,VA 
\ar[d]^{\text{Zhu}} \\
PA\,\,\,\,\,\,  
&
\ar[l]^{\text{cl.limit}}
\,\,\,\,\,\, AA 
}
\end{equation}
(The algebraic structure corresponding to an arbitrary quantum field theory
is still to be understood,
but in the special case of chiral quantum fields
of a 2-dimensional conformal field theory the adequate algebraic structure is a vertex algebra.)
The classical limit associates to a family of associative (resp. vertex) algebras
with a commutative limit a Poisson algebra (resp. Poisson vertex algebra).
Furthermore, the Zhu map associates to a VA (resp. PVA) with an energy operator,
an associative algebra (resp. Poisson algebra), see \cite{Zhu96} (resp. \cite{DSK06}).

The simplest example when all four objects in diagram \eqref{maxi} can be constructed,
is obtained starting with a  finite-dimensional Lie algebra
(or superalgebra) $\mf g$,
with Lie bracket $[\cdot\,,\,\cdot]$,
and with a non-degenerate invariant symmetric bilinear form $(\cdot\,|\,\cdot)$.
We have a family of Lie algebras $\mf g_{\hbar}$, $\hbar\in\mb F$, with underlying space $\mf g$,
and the Lie bracket
\begin{equation}\label{20140401:eq1}
[a,b]_{\hbar}=\hbar[a,b]\,.
\end{equation}
We also have a family of Lie conformal algebras 
$\text{Cur}\,\mf g_{\hbar}=(\mb F[\partial]\otimes\mf g)\oplus\mb FK$,
with the following $\lambda$-bracket:
\begin{equation}\label{20140401:eq2}
[a_\lambda b]_{\hbar}
=
\hbar
\big(
[a,b]+(a|b)K\lambda
\big)
\,\,,\,\,\,\,
[a_\lambda K]=
0
\,\,,\,\,\,\,
\text{ for }\,\, a,b\in\mf g
\,.
\end{equation}
Then, 
the universal enveloping algebra of $\mf g_{\hbar}$
is the family of associative algebras $U(\mf g_{\hbar})$,
and its classical limit is the symmetric algebra $S(\mf g)$,
with the Kirillov-Kostant Poisson bracket
(here the invariant bilinear form plays no role).
Furthermore,
the universal enveloping vertex algebra of $\text{Cur}\,\mf g_{\hbar}$
is the family of vertex algebras $V(\mf g_{\hbar})$,
and its classical limit is the algebra of differential polynomials
$\mc V(\mf g)=S(\mb F[\partial]\mf g)$,
with the PVA $\lambda$-bracket defined by \eqref{20140401:eq2} with $\hbar=1$.
For the definition of the latter structures and the construction of the corresponding Zhu maps,
see \cite{DSK06}.
Thus, 
we get the following example of diagram \eqref{maxi}:
\begin{equation}\label{maxi2}
\UseTips
\xymatrix{
\mc V(\mf g)\,\,\,\,\,\, 
\ar[d]
&
\ar[l]
\,\,\,\,\,\,
V(\mf g_{\hbar})
\ar[d]
\\
S(\mf g)
\,\,\,\,\,\,  
&
\ar[l]
\,\,\,\,\,\, 
U(\mf g_{\hbar})
}
\end{equation}

Now, let $\mf s=\{e,h,f\}$ be an $\mf{sl}_2$-triple in $\mf g$.
Then all the four algebraic structures in diagram \eqref{maxi2}
admit a Hamiltonian reduction.

Recall that a classical finite Hamiltonian reduction (HR)
of a Poisson algebra $\mc P$
is associated to a triple $(\mc P_0,\mc I_0,\varphi)$,
where $\mc P_0$ is a Poisson algebra,
$\mc I_0\subset\mc P_0$ is a Poisson algebra ideal,
and $\varphi:\,\mc P_0\to\mc P$ is a Poisson algebra homomorphism.
The corresponding classical finite HR is the following Poisson algebra:
\begin{equation}\label{20140401:eq3}
\mc W^{\text{fin}}
=
\mc W^{\text{fin}}(\mc P,\mc P_0,\mc I_0,\varphi)
=
\big(\mc P\big/\mc P\varphi(\mc I_0)\big)^{\ad\varphi(\mc P_0)}
\,.
\end{equation}
It is easy to see that the obvious Poisson bracket 
on the commutative associative algebra $\mc W^{\text{fin}}$
is well defined.

Next, recall that a quantum finite HR
of a unital associative algebra $A$
is associated to a triple $(A_0,I_0,\varphi)$,
where $A_0$ is a unital associative algebra,
$I_0\subset A_0$ is its two sided ideal,
and $\varphi:\,A_0\to A$ is a homomorphism of unital associative algebras.
The corresponding finite HR is the following unital associative algebra:
\begin{equation}\label{20140401:eq4}
W^{\text{fin}}
=
W^{\text{fin}}(A,A_0,I_0,\varphi)
=
\big(A\big/A\varphi(I_0)\big)^{\ad\varphi(A_0)}
\,.
\end{equation}
Again, it is easy to see that the obvious product on $W^{\text{fin}}$
is well defined.

The classical affine HR of a PVA $\mc V$
is defined very similarly to the classical finite HR:
\begin{equation}\label{20140401:eq5}
\mc W
=
\mc W(\mc V,\mc V_0,\mc I_0,\varphi)
=
\big(\mc V\big/\mc V\varphi(\mc I_0)\big)^{\ad\varphi(\mc V_0)}
\,,
\end{equation}
where $\mc V_0$ is a PVA,
$\mc I_0\subset\mc V_0$ is its PVA ideal,
and $\varphi:\,\mc V_0\to\mc V$ is a PVA homomorphism.

Given an $\mf{sl}_2$-triple $\mf s$ in $\mf g$,
we can perform all three above Hamiltonian reductions as follows.
Let $\mc P=S(\mf g)$, $A=U(\mf g_{\hbar})$, $\mc V=\mc V(\mf g)$.
Next, let
$$
\mf g=\bigoplus_{j\in\frac12\mb Z}\mf g_j
\,,
$$
be the eigenspace decomposition with respect to $\frac12\ad h$.
Let $\mc P_0=S(\mf g_{>0})\subset S(\mf g)$,
$A_0=U(\mf g_{>0,\hbar})\subset U(\mf g_{\hbar})$,
and $\mc V_0=\mc V(\mf g_{>0})\subset\mc V(\mf g)$,
and let $\varphi$ be the inclusion homomorphism in all three cases.
Furthermore,
let, in the three cases,
$\mc I_0\subset\mc P_0$ be the associative algebra ideal,
$I_0\subset\mc A_0$ be the two sided ideal,
and $\mc I_0\subset\mc V_0$ be the differential algebra ideal,
generated by the set
$$
\big\{m-(f|m)\,\big|\,m\in\mf g_{\geq1}\big\}
\,.
$$
Applying the three Hamiltonian reductions,
we obtain 
the finite classical $\mc W$-algebra $\mc W^{\text{fin}}(\mf g,\mf s)$,
the finite quantum $W$-algebra $W^{\text{fin}}_{\hbar}(\mf g,\mf s)$
(it first appeared in \cite{Pre02}),
and the classical $\mc W$-algebra $\mc W(\mf g,\mf s)$.

Unfortunately,
we don't know of a similar construction of a quantum affine HR for vertex algebras.
One uses instead a more special, cohomological approach,
to construct the family of vertex algebra $W_{\hbar}(\mf g,\mf s)$ \cite{FF90,KW04}.

We thus obtain a Hamiltonian reduction of the whole diagram \eqref{maxi2}
\cite{DSK06}:
\begin{equation}\label{maxi3}
\UseTips
\xymatrix{
\mc W(\mf g,\mf s)
\,\,\,\,\,\, 
\ar[d]
&
\ar[l]
\,\,\,\,\,\,
W_{\hbar}(\mf g,\mf s)
\ar[d]
\\
\mc W^{\text{fin}}(\mf g,\mf s)
\,\,\,\,\,\,  
&
\ar[l]
\,\,\,\,\,\, 
W_{\hbar}^{\text{fin}}(\mf g,\mf s)
}
\end{equation}

In the present paper we study in detail the ``classical'' part of diagram \eqref{maxi3}.
(which we are planning to apply to the ``quantum'' part in a subsequent publication).

The main result of Section \ref{sec:slod} is Theorem \ref{20130521:prop},
which provides an explicit formula for the Poisson bracket of
the classical finite $\mc W$-algebra $\mc W^{\text{fin}}(\mf g,\mf s)$.
This Poisson algebra is viewed here as the algebra of polynomial functions on the Slodowy slice 
$\mc S=f+\mf g^e$
(the equivalence of this definition to the HR definition was proved in \cite{GG02}).
As in \cite{GG02},
we use Weinstein's Theorem (Theorem \ref{20130519:prop} of our paper),
which, in our situation, gives an induced Poisson structure on the submanifold $\mc S$
of the Poisson manifold $\mf g\simeq\mf g^*$,
since $\mc S$ intersects transversally and non-degenerately the symplectic leaves of $\mf g^*$.
Our basic tool is a projection map $\Phi^{(r)}:\,\mf g\to\mf g^e$,
defined by \eqref{20130520:eq8},
for each $r\in\mf g_{\geq0}$.

In Section \ref{sec:2}
we recall the definition of the PVA $\mc W=\mc W(\mf g,\mf s)$
in the form given in \cite{DSKV13},
which is equivalent to the HR definition, but it is more convenient.
Indeed, by this definition, $\mc W$ is a differential subalgebra
of the algebra $\mc V(\mf g_{\leq\frac12})$ of differential polynomials over $\mf g_{\leq\frac12}$.
This allows,
using the decomposition $\mf g_{\leq\frac12}=\mf g^f\oplus[e,\mf g_{\leq-\frac12}]$,
to show in Section \ref{sec:3.2}
that for every $q\in\mf g^f$ there exists a unique element $w(q)\in\mc W$
of the form $w(q)=q+\tilde{r}(q)$,
where $\tilde{r}(q)$ lies in the differential ideal of $\mc V(\mf g_{\leq\frac12})$
generated by $[e,\mf g_{\leq-\frac12}]$,
see Corollary \ref{20140221:cor}.
Due to \cite{DSKV13}, $\mc W$, as a differential algebra,
is isomorphic to the algebra of differential polynomials in the variables $w(q)$, 
where $q$ runs over a basis of $\mf g^f$.
Furthermore, we compute explicitly the term $r(q)$ of $\tilde{r}(q)$ linear in $[e,\mf g_{\leq-\frac12}]$,
see Theorem \ref{20140221:thm1}.

Using these results we are able to compute in Section \ref{sec:3.3}
the explicit PVA $\lambda$-brackets between the generators $w(a)$, $a\in\mf g^f$, of $\mc W$,
see Proposition \ref{20140224:thm} and Theorem \ref{20140304:thm}.

Of course, the same method allows one to obtain an algebraic proof of Theorem \ref{20130521:prop}
on explicit Poisson brackets of $\mc W^{\text{fin}}$.
Alternatively, by the results of \cite[Sec.6]{DSK06},
Theorem \ref{20130521:prop} is obtained from Theorem \ref{20140304:thm}
by putting $\partial=0$ and $\lambda=0$ in formula \eqref{20140304:eq4}.

In \cite{MR14}
they constructed explicitly the generators for the $\mc W$-algebra $\mc W(\mf g,\mf s)$,
where $\mf g$ is a simple Lie algebra of type $A, B, C, D, G$, and $\mf s$ is a principal $\mf{sl}_2$
triple in $\mf g$.
It would be interesting to compare their choice of generators with ours.

In Section \ref{sec:7}
we study the family of Zhu algebras $\Zhu_z \mc W $, 
parametrized by $z\in\mb F$,
and show that it is isomorphic to the Poisson algebras of $z$-deformed Slodowy slice 
$\mc S_z=e+\frac12zh+\mf g^f$, see Theorem \ref{20140319:thm}.
(The standard Zhu algebra corresponds to $z=1$.)
Since, by Theorem \ref{20140318:rem} all these Poisson algebras are isomorphic,
we conclude that the Poisson algebras $\Zhu_z \mc W$ are isomorphic for all values of $z\in\mb F$.
In particular, we have that
\begin{equation}\label{20140401:eq6}
\Zhu{}_1 \mc W\simeq \Zhu{}_0\mc W\,\,\big(\simeq\mc W/\mc W\partial\mc W\big)
\,.
\end{equation}
It is easy to show that the Zhu algebras are isomorphic for all non-zero values of $z$ \cite{DSK06},
but the isomorphism \eqref{20140401:eq6} is quite surprising.

Another surprising corollary of our results is Remark \ref{20140311:rem},
which provides a canonical choice (up to scalar factors)
of generators of the algebra of invariant polynomials on a simple Lie algebra $\mf g$.

In the last section,
we construct the generalized Miura map,
which is an injective homomorphism of the PVA $\mc W(\mf g,\mf s)$
to the tensor product of $\mc V(\mf g_0)$
and the ``fermionic'' PVA $\mc F(\mf g_{\frac12})$.

Throughout the paper, unless otherwise specified, all vector spaces, tensor products etc.,
are defined over a field $\mb F$ of characteristic 0.

\section{Poisson algebra structure of the Slodowy slice}
\label{sec:slod}

\subsection{Poisson structures on manifolds and submanifolds}\label{sec:slod.1}

%
Recall that a \emph{Poisson manifold} $M=M^n$
is endowed with a skewsymmetric 2-vector field $\eta\in\Gamma(\bigwedge^2 TM)$
satisfying the condition that
$[[\eta,\eta]]=0$,
where $[[\cdot\,,\,\cdot]]$ is the Nijenhuis-Schouten bracket on the space 
$\Gamma(\bigwedge^\bullet TM)$ of alternating polyvector fields on $M$.
(It is the unique extension of the usual commutator on the space 
$\Vect(M)=\Gamma(TM)$ of vector fields on $M$,
to a graded Gerstenhaber (=odd Poisson) bracket on the space 
$\Gamma(\bigwedge^\bullet TM)$ of alternating polyvector fields on $M$.)
In local coordinates, the Poisson structure $\eta$ has the form
\begin{equation}\label{20130519:eq6}
\eta(x)=\sum_{i,j=1}^n K(x)_{ij}\frac{\partial}{\partial x_i}\wedge\frac{\partial}{\partial x_j}
\,,
\end{equation}
where $K(x)$ is a skewsymmetric matrix, associated to $\eta$
and the choice of coordinates $\{x_i\}_{i=1}^n$.
%
%

The algebra of functions $C^\infty(M)$ has then a natural structure of a Poisson algebra,
given, in local coordinates, by
\begin{equation}\label{20130518:eq1}
\{f(x),g(x)\}
=
\sum_{i,j=1}^nK_{ij}(x)\frac{\partial f(x)}{\partial x_i}\frac{\partial g(x)}{\partial x_j}
\,\,\,\,\bigg(=
\eta(x)(d_x f\wedge d_x g)
\bigg)
\,.
\end{equation}
In fact, the condition $[[\eta,\eta]]=0$ is equivalent
to the Jacobi identity for the bracket  \eqref{20130518:eq1}.

The Poisson structure $\eta$ defines a map from 1-forms to vector fields,
$$
\eta:\,\Omega^1(M)=\Gamma(T^*M)\to\Gamma(TM)=\Vect(M)\,,
$$
given by the natural pairing of $T_x^*M$ and $T_xM$.
In local coordinates, it is
\begin{equation}\label{20130518:eq3}
\Omega^1(M)\ni\,
\xi(x)=\sum_{i=1}^nF_i(x)dx_i
\mapsto
\eta(\xi)(x)
=\sum_{i,j=1}^nK(x)_{ij}F_j\frac{\partial}{\partial x_i}
\,\in\Vect(M)
\,.
\end{equation}
The \emph{Hamiltonian vector field} $X_h$ associated to the 
function $h(x)\in C^\infty(M)$ is, by definition,
\begin{equation}\label{20130518:eq4}
X_h(x)
=\eta(dh)(x)
=\sum_{i,j=1}^nK(x)_{ij}\frac{\partial h(x)}{\partial x_j}\frac{\partial}{\partial x_i}
=\{h(x),\,\cdot\,\}
\,\in\Vect(M)
\,.
\end{equation}
%
%
Hence, the Poisson structure $\eta$ is uniquely determined by the map 
$X:\,C^\infty(M)\to\Vect(M)$
associating to a smooth function $h\in C^\infty(M)$ the corresponding Hamiltonian vector field $X_h\in\Vect(M)$.


Recall that a Poisson manifold is disjoint union of its \emph{symplectic leaves}: $M=\sqcup_{\alpha}S_\alpha$.
Each symplectic leaf $S\subset M$ is defined by the condition that,
for every $x\in S\subset M$, we have
$$
\eta(x)\big(T^*_xM\big)
=T_x S
\,.
$$
(The inclusion $\subset$ exactly means that $S$ is preserved 
by the integral curves of Hamiltonian vector fields,
while the inclusion $\supset$ means that 
the restriction of the Poisson structure $\eta(x)$ on $S$ is non degenerate for every $x\in S$,
thus making $S$ a symplectic manifold.)
On a symplectic leaf $S$, the symplectic form $\omega(x):\,T_xS\times T_xS\to\mb R$
is easily expressed in terms of the Poisson structure $\eta(x)$ on $M$
(which can be equivalently viewed
as a map $\eta(x):\,T^*_xM\to T_xS\subset T_xM$, cf. \eqref{20130518:eq3},
or as a map $\eta(x):\,T^*_xM\times T^*_xM\to\mb R$).
For $\alpha,\beta\in T^*_xM$, we have:
\begin{equation}\label{20130519:eq8}
\omega(x)(\eta(x)(\alpha),\eta(x)(\beta))=\eta(x)(\alpha,\beta)
\,.
\end{equation}

\begin{theorem}[\cite{Va94}]\label{20130519:prop}
Let $(M,\eta)$ be a Poisson manifold, where $\eta$ is the bi-vector field on $M$ defining the Poisson structure,
and let $N\subset M$ be a submanifold.
Suppose that, for every point $x\in N$, denoting by $(S,\omega)$ the symplectic leaf 
of $M$ through $x$, we have
\begin{enumerate}[(i)]
\item
the restriction of the symplectic form $\omega(x):\,T_xS\times T_xS\to\mb R$
to $T_xN\cap T_xS$ is non-degenerate;
\item
$N$ is transverse to $S$, i.e. $T_xN+T_xS=T_xM$.
\end{enumerate}
Then, 
the Poisson structure on $M$ induces a Poisson structure on $N$,
and the symplectic leaf of $N$ through $x$ is $N\cap S$.
The Poisson structure $\eta^N$ on $N$ is defined as follows.
Given a function $h\in C^\infty(N)$,
we extend it to a function $\tilde h\in C^\infty(M)$,
and we consider the vector field $X_{\tilde h}(x)=\eta(d\tilde{h})(x)\in T_xS\subset T_xM$.
By the non-degeneracy condition (i),
we have the orthogonal decomposition
$T_xS=(T_xN\cap T_xS)\oplus (T_xN\cap T_xS)^{\perp\omega}$.
We then define
$X^N_h(x)=\eta^N(d^Nh)$ as the projection of $X_{\tilde h}(x)$ to $T_xN\cap T_xS$.
\end{theorem}

We shall apply this theorem to a vector space $M$ over $\mb F$ with a polynomial Poisson structure 
(i.e. the matrix $K(x)$ in \eqref{20130519:eq6} is a polynomial function of $x$),
and $N$ an affine subspace of $M$.
Then Theorem \ref{20130519:prop} holds over $\mb F$.
Hence, we get a Poisson bracket on the algebra of polynomial functions on $N$.

\subsection{Example: the Kirillov-Kostant symplectic structure on coadjoint orbits.}\label{sec:slod.2}

Let $\mf g$ be a Lie algebra over $\mb F$.
The Lie bracket $[\cdot\,,\,\cdot]$ on $\mf g$
extends uniquely to a Poisson bracket on the symmetric algebra $S(\mf g)$:
if $\{x_i\}_{i=1}^n$ is a basis of $\mf g$, we have,
for $P,Q\in S(\mf g)$,
\begin{equation}\label{20130519:eq1}
\{P,Q\}=\sum_{i,j=1}^n\frac{\partial P}{\partial x_i}\frac{\partial Q}{\partial x_j}[x_i,x_j]
\,.
\end{equation}
We think of $S(\mf g)$ as the algebra of polynomial functions on $\mf g^*$,
and, therefore, the space $\mf g^*$ is a (algebraic) Poisson manifold.
Let $\{\xi_i\}_{i=1}^n$ be the basis of $\mf g^*$ dual to the given basis of $\mf g$:
$\xi_i(x_j)=\delta_{ij}$.
In coordinates, if we think of $\{x_i\}_{i=1}^n$ as linear functions on $\mf g^*$,
then, by \eqref{20130519:eq1}, 
the Poisson structure $\eta$ evaluated at $\xi\in\mf g^*$ is
\begin{equation}\label{20130519:eq2}
\begin{array}{rcll}
\displaystyle{
\eta(\xi)
}&=&
\displaystyle{
\sum_{i,j=1}^n\xi([x_i,x_j])\frac{\partial}{\partial x_i}\wedge \frac{\partial}{\partial x_j}
}&
\displaystyle{
\in\wedge^2(T_\xi\mf g^*)
} \\
&=&
\displaystyle{
\sum_{i,j=1}^n\xi([x_i,x_j])\xi_i\wedge\xi_j
=\sum_{j=1}^n\ad^*(x_j)(\xi)\wedge\xi_j
}&
\displaystyle{
\in\wedge^2(\mf g^*)
\,,}
\end{array}
\end{equation}
where $\ad^*$ is the coadjoint action of $\mf g$ on $\mf g^*$.
Equivalently, the matrix associated to the Poisson structure $\eta$
in coordinates $\{x_i\}_{i=1}^n$ is
$$
K(\xi)_{ij}=\xi([x_i,x_j])=(\ad^*(x_j)(\xi))(x_i)\,.
$$
%
%
The Poisson structure $\eta$ can be equivalently viewed 
as a map 
$\eta(\xi):\,T^*_\xi\mf g^*\simeq\mf g\to T_\xi\mf g^*\simeq\mf g^*$,
given by
\begin{equation}\label{20130519:eq3}
\eta(\xi)(a)=
\ad^*(a)(\xi)
\,,
\end{equation}
or as a skewsymmetric map
$\eta(\xi):\,T^*_\xi\mf g^*\times T^*_\xi\mf g^*\simeq\mf g\times\mf g\to\mb F$,
given by
\begin{equation}\label{20130519:eq10}
\eta(\xi)(a,b)=
\xi([a,b])\,.
\end{equation}

The symplectic leaves of $\mf g^*$ are given by this well known theorem:
\begin{theorem}[Kirillov-Kostant]\label{thm:kk}
The symplectic leaves of the Poisson manifold $\mf g^*$
are the coadjoint orbits: $S=(\Ad^*G)\xi$.
The symplectic structure 
$\omega(\xi):\, T_\xi S\times T_\xi S\simeq((\ad^*\mf g)\xi)\times((\ad^*\mf g)\xi)\to\mb F$,
on the coadjoint orbit $S=(\Ad^*G)\xi$
is given by
\begin{equation}\label{20130519:eq9}
\omega(\xi)\big((\ad^*a)\xi,(\ad^*b)\xi\big)
=\xi([a,b])
\,.
\end{equation}
\end{theorem}

\subsection{The Slodowy slice
and the classical finite \texorpdfstring{$\mc W$}{W}-algebra}\label{sec:slod.3}

Let $\mf g$ be a simple finite-dimensional Lie algebra,
and let $f\in\mf g$ be a nilpotent element.
By the Jacoboson-Morozov Theorem,
$f$ can be included in an $\mf{sl}_2$-triple $\{e,h=2x,f\}$.
Let $(\cdot\,|\,\cdot)$ be a non-degenerate invariant symmetric bilinear form on $\mf g$,
and let $\psi:\,\mf g\stackrel{\sim}{\to}\mf g^*$ be the isomorphism associated to this 
bilinear form: $\psi(a)=(a|\,\cdot)$.
We also let $\chi=\psi(f)=(f|\,\cdot)\in\mf g^*$.
\begin{definition}[see e.g. \cite{Pre02,GG02}]\label{20130518:def}
The \emph{Slodowy slice} associated to this $\mf{sl}_2$-triple is, by definition,
the following affine space
$$
\mc S=\psi(f+\mf g^e)
=\big\{\chi+\psi(a)\,\big|\,a\in\mf g^e\big\}\subset\mf g^*
\,.
$$
As we state below,
$\mc S$ carries a natural structure of a Poisson manifold,
induced by that of $\mf g^*$.
The \emph{classical finite} $\mc W$-\emph{algebra} $\mc W^{\text{fin}}(\mf g,f)$
can be defined as the Poisson algebra of polynomial functions of the Slodowy slice $\mc S$.
\end{definition}

Let $\xi=\psi(f+r)$, $r\in\mf g^e$, be a point of the Slodowy slice.
The tangent space to the coadjoint orbit $(\Ad^*G)\xi$ at $\xi$ is
\begin{equation}\label{20130519:eq11}
T_\xi(\Ad^*G)\xi\simeq(\ad^*\mf g)\xi=\psi([f+r,\mf g])
\subset\mf g^*\simeq T_\xi\mf g^*
\,,
\end{equation}
while the tangent space to the Slodowy slice at $\xi$ is
\begin{equation}\label{20130519:eq12}
T_\xi(\mc S)
\simeq\psi(\mf g^e)\subset\mf g^*\simeq T_\xi\mf g^*
\,.
\end{equation}
By Theorem \ref{thm:kk},
the symplectic form $\omega(\xi)$ on the coadjoint orbits,
which are the symplectic leaves of $\mf g^*$,
coincides with the following skew-symmetric non-degenerate bilinear form
$\omega_{f+r}$ on $T_\xi((\Ad^*G)\xi)\simeq\psi([f+r,\mf g])\simeq[f+r,\mf g]$:
\begin{equation}\label{20130519:eq14}
\omega_{f+r}([f+r,a],[f+r,b])
=(f+r|[a,b])
\,.
\end{equation}

According to Theorem \ref{20130519:prop},
and in view of \eqref{20130519:eq11} and \eqref{20130519:eq12},
in order to prove that the Slodowy slice $\mc S$ carries a natural Poisson structure
induced by $\mf g^*$,
it suffices to check that, for every $r\in\mf g^e$,
two properties hold:
\begin{enumerate}[(i)]
\item
the restriction of the inner product \eqref{20130519:eq14}
to $[f+r,\mf g]\cap\mf g^e$
is non-degenerate;
\item
$[f+r,\mf g]+\mf g^e=\mf g$.
\end{enumerate}
It was proved by Gan and Ginzburg in the 3rd archive version of \cite{GG02}
that these conditions indeed hold.
Following their argument, we provide here a proof of conditions (i) and (ii) 
for every $r\in\mf g_{\geq0}$ (hence for every $r\in\mf g^e$),
respectively in parts (c) and (d) of the following lemma.
\begin{lemma}\label{20140318:lem}
Let $r\in\mf g_{\geq0}$. Then:
\begin{enumerate}[(a)]
\item
$[f+r,[e,\mf g]]\cap\mf g^e=0$.
\item
The map $\ad(f+r)$ restricts to a bijection
$\ad (f+r):\,[e,\mf g]\to[f+r,[e,\mf g]]$.
\item
If $a\in\mf g$ is such that $[f+r,a]\in\mf g^e$ and
\begin{equation}\label{20140318:eq2}
(a|[f+r,\mf g]\cap\mf g^e)=0
\,,
\end{equation}
then $[f+r,a]=0$.
\item
$[f+r,\mf g]+\mf g^e=\mf g$.
\item
If $f\in\mf g$ is a principal nilpotent element, then
$\mf g=[f+r,\mf g]\oplus\mf g^e$.
\end{enumerate}
\end{lemma}
\begin{proof}
Suppose, by contradiction, 
that $0\neq [f+r,[e,z]]\in\mf g^e$
for some $r\in\mf g_{\geq0}$.
We can expand 
\begin{equation}\label{20140318:eq1}
[e,z]=[e,z_i]+[e,z_{i+\frac12}]+\dots
\,,
\end{equation}
where $0\neq[e,z_i]\in[e,\mf g_i]$, $[e,z_{i+\frac12}]\in[e,\mf g_{i+\frac12}]$, $\dots$.
Since, by assumption, $r\in\mf g_{\geq0}$, 
the component of $[f+r,[e,z]]\in\mf g^e$ in $\mf g^e_{i}$ is
$$
[f,[e,z_i]]\in[f,[e,\mf g_i]]\cap\mf g^e_i\subset[f,\mf g]\cap\mf g^e=0
\,.
$$
(For the last equality, see e.g. \eqref{20140221:eq4} below.)
On the other hand, we know that $\ad f:\,[e,\mf g]\to[f,\mf g]$ is a bijection
(see e.g. \eqref{20130520:eq2} below).
It follows that $[e,z_i]=0$, a contradiction, proving part (a).

For part (b) it suffices to prove that, if $[f+r,[e,z]]=0$, then $[e,z]=0$.
The argument is the same as for part (a):
if we expand $[e,z]$ as in \eqref{20140318:eq1}, with $[e,z_i]\neq0$,
then the component of $[f+r,[e,z]]$ in $\mf g_i$ is $0=[f,[e,z_i]]\neq0$,
a contradiction.

Next, we prove part (c).
By linear algebra, we have
$$
\big([f+r,\mf g]\cap\mf g^e\big)^\perp
=
[f+r,\mf g]^\perp+(\mf g^e)^\perp
=
\ker(\ad(f+r))+[e,\mf g]
\,.
$$
(Here we are using the fact that, for every $\xi\in\mf g$, $T=\ad(\xi)\in\End(\mf g)$ 
is a skewadjoint operator with respect to $(\cdot\,|\,\cdot)$,
hence $(\ker T)^\perp=\im T$ and $(\im T)^\perp=\ker T$.)
Hence, 
condition \eqref{20140318:eq2} is equivalent to
\begin{equation}\label{20140318:eq3}
a=k+[e,z]
\,\,\text{ where }\,\,
k\in\ker(\ad(f+r))
\,\,\text{ and }\,\,
z\in\mf g
\,.
\end{equation}
But then $[f+r,a]=[f+r,[e,z]]\in[f+r,[e,\mf g]]\cap\mf g^e$
and this is zero by part (a).

By part (b), we have that $[f+r,[e,\mf g]]$ has the same dimension as $[e,\mf g]$,
which is equal to the codimension of $\mf g^e$.
Also, by part (a) we have that $[f+r,[e,\mf g]]\cap\mf g^e=0$.
It follows that 
$$
\mf g=[f+r,[e,\mf g]]\oplus\mf g^e\subset[f+r,\mf g]+\mf g^e
\,,
$$
proving part (d).

Finally, let us prove part (e).
For the
principal nilpotent $f$, we have $\dim(\mf g^f)=\rk(\mf g)$,
and it is minimal possible:
$\dim(\mf g^a)\geq\rk(\mf g)$ for every $a\in\mf g$.
By part (d) we also have
$[f+r,\mf g]+\mf g^e=\mf g$.
Hence,
$$
\begin{array}{l}
\displaystyle{
\vphantom{\Big(}
\dim(\mf g)
=
\dim\big([f+r,\mf g]+\mf g^e\big)
\leq
\dim([f+r,\mf g])+\dim(\mf g^e)
} \\
\displaystyle{
\vphantom{\Big(}
=
\dim(\mf g)-\dim(\mf g^{f+r},\mf g])+\rk(\mf g)
\leq
\dim(\mf g)
\,.}
\end{array}
$$
Hence, all the inequalities above must be equalities,
proving (e).
\end{proof}
The algebra of polynomial functions on the Slodowy slice $\mc S$ with the obtained Poisson structure
is called the \emph{classical finite} $\mc W$-\emph{algebra},
and it is denoted by $\mc W^{\text{fin}}(\mf g,f)$.

\subsection{Setup and notation}\label{slod.4}

Let, as before, $\mf g$ be a simple Lie algebra with a non-degenerate symmetric invariant bilinear form $(\cdot\,|\,\cdot)$,
and let $\{f,2x,e\}\subset\mf g$ be an $\mf{sl}_2$-triple in $\mf g$.
We have the corresponding $\ad x$-eigenspace decomposition
$$
\mf g=\bigoplus_{k\in\frac{1}{2}\mb Z}\mf g_{k}
\,\,\text{ where }\,\,
\mf g_k=\big\{a\in\mf g\,\big|\,[x,a]=ka\big\}
\,.
$$
Clearly, $f\in\mf g_{-1}$, $x\in\mf g_{0}$ and $e\in\mf g_{1}$.
We let $d$ be the \emph{depth} of the grading, i.e. the maximal eigenvalue of $\ad x$.

By representation theory of $\mf{sl}_2$, the Lie algebra $\mf g$ admits the direct sum decompositions
\begin{equation}\label{20140221:eq4}
\mf g
=\mf g^f\oplus[e,\mf g]
=\mf g^e\oplus[f,\mf g]
\,.
\end{equation}
They are dual to each other, in the sense that $\mf g^f\perp[f,\mf g]$ and $[e,\mf g]\perp\mf g^e$.
For $a\in\mf g$, we denote by $a^\sharp=\pi_{\mf g^f}(a)\in\mf g^f$ its component in $\mf g^f$
with respect to the first decomposition in \eqref{20140221:eq4}.
Note that, since $[e,\mf g]$ is orthogonal to $\mf g^e$,
the spaces $\mf g^f$ and $\mf g^e$ are non-degenerately paired by $(\cdot\,|\,\cdot)$.

Next, we choose a basis of $\mf g$ as follows.
Let $\{q_j\}_{j\in J^f}$ be a basis of $\mf g^f$ consisting of $\ad x$-eigenvectors,
and let  $\{q^j\}_{j\in J^f}$ be the the dual basis of $\mf g^e$.
For $j\in J^f$,
we let $\delta(j)\in\frac12\mb Z$ be the $\ad x$-eigenvalue of $q^j$,
so that
\begin{equation}\label{20130520:eq5}
[x,q_j]=-\delta(j)q_j
\,\,,\,\,\,\,
[x,q^j]=\delta(j)q^j
\,.
\end{equation}
For $k\in\frac12\mb Z_+$
we also let $J^f_{-k}=\{i\in J^f\,|\,\delta(i)=k\}\subset J^f$,
so that $\{q_j\}_{j\in J^f_{-k}}$ is a basis of $\mf g^f_{-k}$,
and $\{q^j\}_{j\in J^f_{-k}}$ is the dual basis of $\mf g^e_{k}$.
By representation theory of $\mf{sl}_2$,
we get a basis of $\mf g$ consisting of the following elements:
\begin{equation}\label{20140221:eq1}
q^j_n=(\ad f)^nq^j
\,\,\text{ where }\,\,
n\in\{0,\dots,2\delta(j)\}
\,\,,\,\,\,\,
j\in J^f
\,.
\end{equation}
This basis consists of $\ad x$-eigenvectors,
and, for $k\in\frac12\mb Z$ such that $-d\leq k\leq d$,
the corresponding basis of $\mf g_k\subset\mf g$ is 
$\{q^j_{n}\}_{(j,n)\in J_{-k}}$,
where $J_{-k}$ is the following index set
\begin{equation}\label{20140221:eq5}
J_{-k}
=
\Big\{
(j,n)\in J^f\times\mb Z_+\,\Big|\,
\delta(j)-|k|\in\mb Z_+,\,n=\delta(j)-k
\Big\}
\,.
\end{equation}
The union of all these index sets is the index set for the basis of $\mf g$:
\begin{equation}\label{20140221:eq6}
J
=
\bigsqcup_{h\in\frac12\mb Z}J_h
=
\Big\{
(j,n)\,\Big|\,
j\in J^f,\,n\in\{0,\dots,2\delta(j)\}
\Big\}
\,.
\end{equation}

The corresponding dual basis of $\mf g$ is given by the following lemma.
\begin{lemma}\label{20140221:lem1}
For 
$i,j\in J^f$ and $m,n\in\mb Z_+$, we have
\begin{equation}\label{20140221:eq2}
((\ad e)^nq_j|(\ad f)^mq^i)
=
(-1)^n(n!)^2\binom{2\delta(j)}{n}
\delta_{i,j}\delta_{m,n}
\,.
\end{equation}
Hence,
the basis of $\mf g$ dual to \eqref{20140221:eq1} is given by ($(j,n)\in J$):
\begin{equation}\label{20140221:eq3}
q_j^n
=
\frac{(-1)^n}{(n!)^2\binom{2\delta(j)}{n}}
(\ad e)^nq_j
\,.
\end{equation}
\end{lemma}
\begin{proof}
Equation \eqref{20140221:eq2} is easily proved by induction on $n$, 
using the invariance of the bilinear form.
\end{proof}
We will also need the following simple facts about the $\mf{sl}_2$ action
on the dual bases $\{q^n_j\}$ and $\{q^j_n\}$.
\begin{lemma}\label{20140304:lem2}
For $j\in J^f$ and $n\in\{0,1,\dots,2\delta(j)\}$, we have
\begin{enumerate}[(i)]
\item
$[f,q^n_j]=-q^{n-1}_j$,
\item
$[e,q^n_j]=-(n+1)(2\delta(j)-n)q^{n+1}_j$,
\item
$[f,q^j_n]=q^j_{n+1}$,
\item
$[e,q^j_n]=n(2\delta(j)-n+1)q^j_{n-1}$.
\end{enumerate}
In the above equations we let $q_j^{-1}=q^j_{-1}=q_j^{2\delta(j)+1}=q^j_{2\delta(j)+1}=0$.
\end{lemma}
\begin{proof}
All formulas are easily proved by induction,
using the formulas \eqref{20140221:eq1} and \eqref{20140221:eq3}
for $q^j_n$ and $q^n_j$ respectively.
\end{proof}

Clearly, the bases \eqref{20140221:eq1} and \eqref{20140221:eq3} 
are compatible with the direct sum decompositions \eqref{20140221:eq4}.
In fact, we can write the corresponding projections
$\pi_{\mf g^f}$, $\pi_{[e,\mf g]}=1-\pi_{\mf g^f}$,
$\pi_{\mf g^e}$, and $\pi_{[f,\mf g]}=1-\pi_{\mf g^e}$,
in terms of these bases:
\begin{equation}\label{20130520:eq1}
\begin{array}{l}
\displaystyle{
\vphantom{Big(}
a^\sharp = \pi_{\mf g^f}(a)=\sum_{j\in J^f}(a|q^j)q_j
\,\,,\,\,\,\,
\pi_{[e,\mf g]}(a)=\sum_{j\in J^f}\sum_{n=1}^{2\delta(j)}(a|q^j_n)q_j^n
\,,} \\
\displaystyle{
\vphantom{Big(}
\pi_{\mf g^e}(a)=\sum_{j\in J^f}(a|q_j)q^j
\,\,,\,\,\,\,
\pi_{[f,\mf g]}(a)=\sum_{j\in J^f}\sum_{n=1}^{2\delta(j)}(a|q_j^n)q^j_n
\,.}
\end{array}
\end{equation}

\subsection{Preliminary results}\label{slod.5}

Due to the decomposition \eqref{20140221:eq4},
the adjoint action of $f$ restricts to a bijective map
\begin{equation}\label{20130520:eq2}
\ad f\,:\,\,[e,\mf g]\stackrel{\sim}{\longrightarrow}[f,\mf g]
\,,
\end{equation}
and we denote by $(\ad f)^{-1}:\,[f,\mf g]\to[e,\mf g]$ the inverse map.
Therefore, we have the following well defined map
$(\ad f)^{-1}\circ\pi_{[f,\mf g]}:\,\mf g\to[e,\mf g]$,
which is obviously surjective and with kernel $\mf g^e$.
We have an explicit formula for it, in terms of the bases \eqref{20140221:eq1}-\eqref{20140221:eq3},
using the last completeness relation in \eqref{20130520:eq1}:
\begin{equation}\label{20130520:eq3}
(\ad f)^{-1}\circ\pi_{[f,\mf g]}(a)
=
\sum_{j\in J^f}\sum_{n=0}^{2\delta(j)-1}(a|q_j^{n+1})q^j_n
\,.
\end{equation}
\begin{lemma}\label{20130520:lem1}
Let $r\in\mf g_{\geq0}$.
Then the map $(\ad r)\circ(\ad f)^{-1}\circ\pi_{[f,\mf g]}:\,\mf g\to\mf g$ is nilpotent.
In fact, it is zero when raised to a power greater than twice the depth $d$ of $\mf g$.
\end{lemma}
\begin{proof}
Clearly, $\pi_{[f,\mf g]}$ is homogeneous with respect to the $\ad x$-eigenspace decomposition,
and it does not change the $\ad x$-eigenvalues,
the map $(\ad f)^{-1}$ is also homogeneous with respect to the $\ad x$-eigenspace decomposition,
and it increases the $\ad x$-eigenvalues by $1$,
while the map $\ad r$ is not homogeneous, but it does not decrease the $\ad x$-eigenvalues 
(since, by assumption, $r\in\mf g_{\geq0}$).
The claim follows.
\end{proof}
For $r\in\mf g_{\geq0}$, consider the map
$\Phi^{(r)}:\,\mf g\to\mf g$, given by
\begin{equation}\label{20130520:eq8}
\begin{array}{l}
\displaystyle{
\Phi^{(r)}
=
\pi_{\mf g^e}\circ
\big(1+(\ad r)\circ(\ad f)^{-1}\circ\pi_{[f,\mf g]}\big)^{-1}
} \\
\displaystyle{
=
\pi_{\mf g^e}\circ\sum_{t=0}^{2d}\big(-(\ad r)\circ(\ad f)^{-1}\circ\pi_{[f,\mf g]}\big)^t
\,.}
\end{array}
\end{equation}
Note that, thanks to Lemma \ref{20130520:lem1},
we can replace $d$ by $\infty$ in the above summation.
Associated to $r\in\mf g_{\geq0}$, we also introduce the following vector space :
\begin{equation}\label{20130520:eq9b}
V^{(r)}:=\Big\{[f+r,a]\,\Big|\,a\in\mf g,\,\big(a|[f+r,\mf g]\cap\mf g^e\big)=0\Big\}\,.
\end{equation}
\begin{lemma}\label{20130520:lem2}
For every $r\in\mf g_{\geq0}$ we have:
\begin{enumerate}[(a)]
\item
$\Phi^{(r)}(a)\in\mf g^e$ for every $a\in\mf g$.
\item
$\Phi^{(r)}(a)=a$
for every $a\in\mf g^e$.
\item For every $a\in\mf g$, we have
\begin{equation}\label{20130520:eq10}
a-\Phi^{(r)}(a)
=
\Big[
f+r
,
(\ad f)^{-1}\circ\pi_{[f,\mf g]}\circ\big(1+(\ad r)\circ(\ad f)^{-1}\circ\pi_{[f,\mf g]}\big)^{-1}(a)
\Big]
\,.
\end{equation}
\item
$\ker(\Phi^{(r)})=\Span\{a-\Phi^{(r)}(a)\,|\,a\in\mf g\}\subset[f+r,\mf g]$.
\item
$V^{(r)}\cap\mf g^e=0$.
\item
$\ker(\Phi^{(r)})= V^{(r)}$.
\item
We have the direct sum decomposition
\begin{equation}\label{20130520:eq9a}
\mf g=\mf g^e\oplus V^{(r)}
\,,
\end{equation}
and 
$\Phi^{(r)}$ 
is the projection onto $\mf g^e$ with kernel $V^{(r)}$.
\item
We have the direct sum decomposition
\begin{equation}\label{20130520:eq9}
[f+r,\mf g]=\big([f+r,\mf g]\cap\mf g^e\big)\oplus V^{(r)}
\,,
\end{equation}
and 
$\Phi^{(r)}\big|_{[f+r,\mf g]}$ 
is the projection onto $[f+r,\mf g]\cap\mf g^e$ with kernel $V^{(r)}$.
\end{enumerate}
\end{lemma}
\begin{proof}
Part (a) is obvious.
Part (b) is clear, since $\pi_{[f,\mf g]}$ acts trivially on $\mf g^e$ due to the decomposition \eqref{20140221:eq4}.
For part (c), we have
$$
\begin{array}{l}
\displaystyle{
\vphantom{\Big(}
a-\Phi^{(r)}(a)
=
a-
\pi_{\mf g^e}\circ
\big(1+(\ad r)\circ(\ad f)^{-1}\circ\pi_{[f,\mf g]}\big)^{-1}(a)
} \\
\displaystyle{
\vphantom{\Big(}
=
a
-\big(1+(\ad r)\circ(\ad f)^{-1}\circ\pi_{[f,\mf g]}\big)^{-1}(a)
} \\
\displaystyle{
\vphantom{\Big(}
\,\,\,\,\,\,\,\,\,
+\pi_{[f,\mf g]}\circ
\big(1+(\ad r)\circ(\ad f)^{-1}\circ\pi_{[f,\mf g]}\big)^{-1}(a)
} \\
\displaystyle{
\vphantom{\Big(}
=
(\ad r)\circ(\ad f)^{-1}\circ\pi_{[f,\mf g]}\circ\big(1+(\ad r)\circ(\ad f)^{-1}\circ\pi_{[f,\mf g]}\big)^{-1}(a)
} \\
\displaystyle{
\vphantom{\Big(}
\,\,\,\,\,\,\,\,\,
+\pi_{[f,\mf g]}\circ
\big(1+(\ad r)\circ(\ad f)^{-1}\circ\pi_{[f,\mf g]}\big)^{-1}(a)
} \\
\displaystyle{
\vphantom{\Big(}
=
\Big[
f+r
,
(\ad f)^{-1}\circ\pi_{[f,\mf g]}\circ\big(1+(\ad r)\circ(\ad f)^{-1}\circ\pi_{[f,\mf g]}\big)^{-1}(a)
\Big]
\,.}
\end{array}
$$
The equality in part (d) follows from (a) and (b),
and the inclusion follows from part (c).
Part (e) is the same as Lemma \ref{20140318:lem}(c).
The inclusion $\ker\Phi^{(r)}\subset V^{(r)}$ 
follows from (d) and (c), and the observation that 
the map $(\ad f)^{-1}\circ\pi_{[f,\mf g]}$ has values in $[e,\mf g]$,
which is orthogonal to $\mf g^e$.
By parts (a), (b), (e) and the inclusion $\ker\Phi^{(r)}\subset V^{(r)}$ 
we have 
$$
\mf g=\mf g^e\oplus\ker(\Phi^{(r)})\subset \mf g^e+V^{(r)}=\mf g^e\oplus V^{(r)}
\,.
$$
Part (f) follows.
Finally, part (g) is an immediate consequence of (a)--(f),
and part (h) follows from (g) and the fact that $V^{(r)}\subset[f+r,\mf g]$.
\end{proof}

\subsection{Explicit formula for the Poisson structure of the Slodowy slice}\label{slod.6}

%
Note that we can identify the ``dual space'' to $\mc S=\psi(f+\mf g^e)\simeq\mf g^e$
as $\mf g^f$ (via the non-degenerate pairing of $\mf g^e$ and $\mf g^f$).
Hence, the classical finite $\mc W$-algebra is, as a commutative associative algebra,
\begin{equation}\label{20130519:eq15}
\mc W^{\text{fin}}(\mf g,f)\simeq S(\mf g^f)\,.
\end{equation}
A polynomial function $P$ on $\mc S$
can be identified with an element $P\in S(\mf g^f)$
which can viewed as an element of $S(\mf g)$,
and therefore it can be considered as a polynomial function on $\mf g^*$.
Clearly, the restriction of it to $\mc S$ coincides with the polynomial function $P$
we started with 
(we are using the fact that $P(f+r)=P(r)$, since $(f|\mf g^f)=0$).
The Poisson structure $\eta^{\mc S}$ on $\mc S$ 
is given by Theorem \ref{20130519:prop}.
Fix $\xi=\psi(f+r)\in\mc S$, where $r\in\mf g^e$.
We have the $\omega_{f+r}$-orthogonal decomposition
(cf. \eqref{20130519:eq11}, \eqref{20130519:eq12} and \eqref{20130519:eq14})
\begin{equation}\label{20130519:eq18}
T_\xi\mc S\simeq[f+r,\mf g]=(\mf g^e\cap[f+r,\mf g])\oplus V^{(r)}\,,
\end{equation}
where $V^{(r)}$ is as in \eqref{20130520:eq9b}.
By Lemma \ref{20130520:lem2}(h),
the projection onto the first component
is the map $\Phi^{(r)}|_{[f+r,\mf g]}:\,[f+r,\mf g]\to\mf g^e\cap[f+r,\mf g]$
given by \eqref{20130520:eq8}.
If we consider $\eta^{\mc S}(\xi)$ as a map
$\eta^{\mc S}(\xi):\,T^*_\xi\mc S\simeq\mf g^f\to T_\xi\mc S\simeq\mf g^e$,
we have that ($q\in\mf g^f$):
\begin{equation}\label{20130519:eq16}
\eta^{\mc S}(\xi)(q)
=\Phi^{(r)}([q,f+r])
=\Phi^{(r)}([q,r])
\,.
\end{equation}
Equivalently, we can consider $\eta^{\mc S}(\xi)$ as a skewsymmetric map
$\eta^{\mc S}(\xi):\, \mf g^f\times\mf g^f\to\mb F$,
given by ($p,q\in\mf g^f$)
\begin{equation}\label{20130519:eq17}
\eta^{\mc S}(\xi)(p,q)
=
\big(p|\Phi^{(r)}([q,r])\big)
\,.
\end{equation}
To get the corresponding Poisson bracket on $\mc W^{\text{fin}}(\mf g,f)=S(\mf g^f)$,
we should 
write the RHS of \eqref{20130519:eq16}
as a polynomial function $P^{p,q}(q_1,\dots,q_\ell)$ ($\ell=\dim\mf g^f$) in the elements of $\mf g^f$,
i.e.
$$
\big(p|\Phi^{(r)}([q,r])\big)=P^{p,q}((q_1|r),\dots,(q_\ell|r))\,,
$$ 
for all $r\in\mf g^e$.
Then, this polynomial gives the corresponding Poisson bracket among
generators of the classical finite $\mc W$-algebra:
\begin{equation}\label{20130519:eq19}
\{p,q\}_{\mc S}=
P^{p,q}(q_1,\dots,q_\ell)
\in S(\mf g^f)
\,.
\end{equation}
\begin{example}\label{20130519:ex}
If $q\in\mf g^f_0=\mf g^e_0$, 
then $[q,f+r]=[q,r]\in\mf g^e\cap[f+r,\mf g]$,
and therefore $\Phi^{(r)}([q,r])=[q,r]$.
It follows that $P^{p,q}(r)=(p|\Phi^{(r)}([q,r]))=(p|[q,r])=([p,q]|r)$,
i.e. $P^{p,q}=[p,q]\in\mf g^f\subset S(\mf g^f)$.
By skewsymmetry, if $p\in\mf g^f_0$, we also have $P^{p,q}=[p,q]$.
In conclusion, 
$$
\{p,q\}_{\mc S}=[p,q]\in\mf g^f
\,\,\text{ if } p,q\in\mf g^f \text{ and either } p \text{ or } q \in\mf g^f_0
\,.
$$
\end{example}
\begin{example}\label{20140318:ex}
If $f\in\mf g$ is a principal nilpotent element and $r\in\mf g^e$,
then, by Lemma \ref{20140318:lem}(e),
we have $\mf g=[f+r,\mf g]\oplus\mf g^e$.
It follows by \eqref{20130520:eq9b} that $V^{(r)}=[f+r,\mf g]$.
On the other hand, for $q\in\mf g^f$, we have $[r,q]=[f+r,q]\in V^{(r)}=\ker(\Phi^{(r)})$.
Hence, $\Phi^{(r)}([q,r])=0$.
It follows that
$\{p,q\}_{\mc S}=0$ for every $p,q\in\mf g^f$:
the Poisson bracket is identically zero in this case,
which is a well known result of Kostant \cite{Kos78}.
\end{example}
\begin{theorem}\label{20130521:prop}
The general formula for the Poisson bracket on the $\mc W$-algebra 
$\mc W^{\text{fin}}(\mf g,f)$ is ($p,q\in\mf g^f$):
\begin{equation}\label{20140319:eq7}
\{p,q\}_{\mc S}=
[p,q]+
\sum_{t=1}^\infty
\sum_{j_1,\dots,j_t\in J^f}
\sum_{n_1=0}^{2\delta(j_1)-1}
\dots
\sum_{n_t=0}^{2\delta(j_t)-1}
[p,q^{j_1}_{n_1}]^\sharp
[q^{n_1+1}_{j_1},q^{j_2}_{n_2}]^\sharp
\dots
[q^{n_t+1}_{j_t},q]^\sharp
\,,
\end{equation}
where $a^\sharp=\pi_{\mf g^f}(a)$.
\end{theorem}
\begin{proof}
By \eqref{20130519:eq19} and \eqref{20130520:eq8}, we have
$\{p,q\}_{\mc S}=\sum_{t=0}^\infty\{p,q\}_{\mc S}^{(t)}$, where
$$
\{p,q\}_{\mc S}^{(t)}(r)=
\Big(p\Big|
\pi_{\mf g^e}\circ\big(-(\ad r)\circ(\ad f)^{-1}\circ\pi_{[f,\mf g]}\big)^t([q,r])
\Big)
\,.
$$
Note that we can remove $\pi_{\mf g^e}$ since, by assumption, $p\in\mf g^f$.
For $t=0$ we get $\{p,q\}_{\mc S}^{(0)}(r)=(p|[q,r])$,
i.e. $\{p,q\}_{\mc S}^{(0)}=[p,q]$.
For $t\geq1$ we have
$$
\{p,q\}_{\mc S}^{(t)}(r)
=
\Big(r\Big|
\Big[p,
(\ad f)^{-1}\circ\pi_{[f,\mf g]}\circ
\big(-(\ad r)\circ(\ad f)^{-1}\circ\pi_{[f,\mf g]}\big)^{t-1}([q,r])
\Big]\Big)
\,.
$$
We can use equation \eqref{20130520:eq3} to rewrite the above equation as
$$
\{p,q\}_{\mc S}^{(t)}(r)
=
\sum_{j_1\in J^f}\sum_{n_1=0}^{2\delta(j_1)-1}
([p,q^{j_1}_{n_1}]|r)
\Big(
q_{j_1}^{n_1+1}
\Big|
\big(-(\ad r)\circ(\ad f)^{-1}\circ\pi_{[f,\mf g]}\big)^{t-1}([q,r])
\Big)
\,.
$$
The general formula follows by an easy induction on $t$.
\end{proof}
\begin{remark}\label{20140310:rem1}
Formula \eqref{20140319:eq7} appeared for the first time in \cite{DS13}.
\end{remark}
\begin{remark}\label{20140310:rem2}
The Poisson bracket \eqref{20140319:eq7} is homogeneous
with respect to the conformal weight (cf. Section \ref{sec:2.2}),
which coincides with the so-called Kazhdan grading.
In fact, the Poisson algebra $\mb F[\mc S]$ can be viewed as the associated graded
(or ``classical limit'') of the quantum finite $W$-algebra $W^{\text{fin}}(\mf g,f)$
with respect to the Kazhdan filtration, \cite{GG02}.
(Here $W$, as opposed to $\mc W$, refers to ``quantum'' $W$-algebras.)
\end{remark}

\subsection{Twisted Slodowy slice}\label{slod.7}

We can consider also the following ``twisted'' Slodowy slice 
$$
\mc S_z=\psi(f+zx+\mf g^e)
\,\,,\,\,\,\,
z\in\mb F
\,.
$$
Since $x\in\mf g_0$,
and all the preliminary results from Sections \ref{sec:slod.3} and \ref{slod.5}
hold for $r\in\mf g_{\geq0}$ (not only for $r\in\mf g^e$),
we can repeat the same arguments that lead to Theorem \ref{20130521:prop}
(replacing everywhere $r$ by $zx+r$)
to get the Poisson algebra structure on the algebra of polynomial functions 
$\mc W^{\text{fin}}_z(\mf g,f)\simeq S(\mf g^f)$ on $\mc S_z$.
We thus get the following ``twisted'' analogue of Theorem \ref{20130521:prop}:
\begin{theorem}\label{20140318:rem}
The general formula for the Poisson bracket on generators
of the $z$-twisted classical finite $\mc W$-algebra $\mc W_z(\mf g,f)$
is as follows ($p,q\in\mf g^f$):
\begin{equation}\label{20140319:eq6}
\begin{array}{l}
\displaystyle{
\vphantom{\Big(}
\{p,q\}_{\mc S_z}
} \\
\displaystyle{
\vphantom{\Big(}
=
[p,q]+z(x|[p,q])+
\sum_{t=1}^\infty
\sum_{j_1,\dots,j_t\in J^f}
\sum_{n_1=0}^{2\delta(j_1)-1}
\dots
\sum_{n_t=0}^{2\delta(j_t)-1}
\big(
[p,q^{j_1}_{n_1}]^\sharp
+z(x|[p,q^{j_1}_{n_1}])
\big)
} \\
\displaystyle{
\vphantom{\Big(}
\big(
[q^{n_1+1}_{j_1},q^{j_2}_{n_2}]^\sharp
+z(x|[q^{n_1+1}_{j_1},q^{j_2}_{n_2}])
\big)
\dots
\big(
[q^{n_t+1}_{j_t},q]^\sharp
+z(x|[q^{n_t+1}_{j_t},q])
\big)
\,.}
\end{array}
\end{equation}
\end{theorem}
Note that all the $z$-twisted Poisson algebras $\mc W_z^{\text{fin}}(\mf g,f)$
are isomorphic for every $z\in\mb F$.
Indeed (as Pasha Etingof pointed out), 
the automorphism $e^{\frac12 z\ad^*(e)}$ of $\mf g^*$
maps $\mc S=\psi(f+\mf g^e)$ to $\mc S_z=\psi(f+zx+\mf g^e)$,
hence it induces a Poisson algebra isomorphism 
$\mc W^{\text{fin}}(\mf g,f)\stackrel{\sim}{\to}\mc W_z^{\text{fin}}(\mf g,f)$.
This algebra isomorphism is obtained as pullback of the map of Poisson manifolds
$\phi:\,\mc S_z\stackrel{\sim}{\to}\mc S$, given by 
$$
\phi:\,f+zx+r\mapsto e^{-\frac12z\ad e}(f+zx+r)=f+r+\frac{z^2}{4}e
\,.
$$
Hence, it maps the generators $q\in\mf g^f$ 
(viewed as a linear function $(q|\,\cdot)$ on $\mc S\simeq\mf g^e$)
to 
\begin{equation}\label{20140402:eq1}
\phi^*(q)=q+\frac{z^2}{4}(q|e)
\,.
\end{equation}
As a consequence, we get the induced Poisson algebra isomorphism
$\phi^*:\,\mc W(\mf g,f)\simeq S(\mf g^f)\stackrel{\sim}{\longrightarrow}\mc W_z(\mf g,f)\simeq S(\mf g^f)$.
In other words, we have the following
\begin{corollary}\label{20140318:remb}
The map 
$\phi^*:\,\mc W(\mf g,f)\simeq S(\mf g^f)\stackrel{\sim}{\longrightarrow}
\mc W_z(\mf g,f)\simeq S(\mf g^f)$
defined, on generators, by \eqref{20140402:eq1},
is a Poisson algebra isomorphism
from the Poisson algebra $\mc W(\mf g,f)\simeq S(\mf g^f)$ with Poisson bracket \eqref{20140319:eq7}
to the Poisson algebra $\mc W_z(\mf g,f)\simeq S(\mf g^f)$ with Poisson bracket \eqref{20140319:eq6}.
In other words, equation \eqref{20140319:eq6} is unchanged
if we replace $zx$ by $\frac{z^2}{4}e$.
\end{corollary}

\section{Classical affine \texorpdfstring{$\mc W$}{W}-algebras
\texorpdfstring{$\mc W(\mf g,f)$}{W(g,f)}}\label{sec:2}

In this section we recall the definition of classical affine $\mc W$-algebras $\mc W(\mf g,f)$ in the language 
of Poisson vertex algebras, following \cite{DSKV13}
(which is a development of \cite{DS85}).
We refer to \cite{BDSK09} for the definition of Poisson vertex algebras (PVA)
and their basic properties.
We shall use the setup and notation of Section \ref{slod.4}.

\subsection{Construction of the classical affine \texorpdfstring{$\mc W$}{W}-algebra}
\label{sec:2.1}

Let $\mf g$ be a simple finite-dimensional Lie algebra over the field $\mb F$
with a non-degenerate symmetric invariant bilinear form $(\cdot\,|\,\cdot)$.
Given an element $s\in\mf g$, we have a PVA structure on
the algebra of differential polynomials $\mc V(\mf g)=S(\mb F[\partial]\mf g)$,
with $\lambda$-bracket given on generators by 
\begin{equation}\label{lambda}
\{a_\lambda b\}_z=[a,b]+(a| b)\lambda+z(s|[a,b]),
\qquad a,b\in\mf g\,,
\end{equation}
and extended to $\mc V(\mf g)$ by the sesquilinearity axioms and the Leibniz rules.
Here $z$ is an element of the field $\mb F$.

We shall assume that $s$ lies in $\mf g_d$.
In this case the $\mb F[\partial]$-submodule
$\mb F[\partial]\mf g_{\geq\frac12}\subset\mc V(\mf g)$ 
is a Lie conformal subalgebra 
with the $\lambda$-bracket $\{a_\lambda b\}_z=[a,b]$, $a,b\in\mf g_{\geq\frac12}$
(it is independent of $z$, since $s$ commutes with $\mf g_{\geq\frac12}$).
Consider the differential subalgebra
$\mc V(\mf g_{\leq\frac12})=S(\mb F[\partial]\mf g_{\leq\frac12})$ of $\mc V(\mf g)$,
and denote by $\rho:\,\mc V(\mf g)\twoheadrightarrow\mc V(\mf g_{\leq\frac12})$,
the differential algebra homomorphism defined on generators by
\begin{equation}\label{rho}
\rho(a)=\pi_{\leq\frac12}(a)+(f| a),
\qquad a\in\mf g\,,
\end{equation}
where $\pi_{\leq\frac12}:\,\mf g\to\mf g_{\leq\frac12}$ denotes 
the projection with kernel $\mf g_{\geq1}$.
Recall from \cite{DSKV13} that
we have a representation of the Lie conformal algebra $\mb F[\partial]\mf g_{\geq\frac12}$ 
on the differential subalgebra $\mc V(\mf g_{\leq\frac12})\subset\mc V(\mf g)$ given by
($a\in\mf g_{\geq\frac12}$, $g\in\mc V(\mf g_{\leq\frac12})$):
\begin{equation}\label{20120511:eq1}
a\,^\rho_\lambda\,g=\rho\{a_\lambda g\}_z
\end{equation}
(note that the RHS is independent of $z$ since, by assumption, $s\in Z(\mf g_{\geq\frac12})$).

The \emph{classical} $\mc W$-\emph{algebra} is, by definition,
the differential algebra
\begin{equation}\label{20120511:eq2}
\mc W=\mc W(\mf g,f)
=\big\{g\in\mc V(\mf g_{\leq\frac12})\,\big|\,a\,^\rho_\lambda\,g=0\,\text{ for all }a\in\mf g_{\geq\frac12}\}\,,
\end{equation}
endowed with the following PVA $\lambda$-bracket
\begin{equation}\label{20120511:eq3}
\{g_\lambda h\}_{z,\rho}=\rho\{g_\lambda h\}_z,
\qquad g,h\in\mc W\,.
\end{equation}
\begin{theorem}[\cite{DSKV13}]\phantomsection\label{daniele1}
\begin{enumerate}[(a)]
\item
Equation \eqref{20120511:eq1} defines a representation 
of the Lie conformal algebra $\mb F[\partial]\mf g_{\geq\frac12}$ on $\mc V(\mf g_{\leq\frac12})$
by derivations
(i.e. $a\,^\rho_\lambda\,(gh)=(a\,^\rho_\lambda\,g)h+(a\,^\rho_\lambda\,h)g$).
\item
$\mc W\subset\mc V(\mf g_{\leq\frac12})$ is a differential subalgebra.
\item
We have
$\rho\{g_\lambda\rho(h)\}_z=\rho\{g_\lambda h\}_z$,
and $\rho\{\rho(h)_\lambda g\}_z=\rho\{h_\lambda g\}_z$
for every $g\in\mc W$ and $h\in\mc V(\mf g)$.
\item
For every $g,h\in\mc W$, we have $\rho\{g_\lambda h\}_z\in\mc W[\lambda]$.
\item
The $\lambda$-bracket 
$\{\cdot\,_\lambda\,\cdot\}_{z,\rho}:\,\mc W\otimes\mc W\to\mc W[\lambda]$
given by \eqref{20120511:eq3} defines a PVA structure on $\mc W$.
\end{enumerate}
\end{theorem}

\subsection{Structure Theorem for classical affine \texorpdfstring{$\mc W$}{W}-algebras}
\label{sec:2.2}

In the algebra of differential polynomials $\mc V(\mf g_{\leq\frac12})$
we introduce the grading by \emph{conformal weight},
denoted by $\Delta\in\frac12\mb Z$, defined as follows.
For $a\in\mf g$ such that $[x,a]=\delta(a)a$, we let $\Delta(a)=1-\delta(a)$.
For a monomial $g=a_1^{(m_1)}\dots a_s^{(m_s)}$,
product of derivatives of $\ad x$ eigenvectors $a_i\in\mf g_{\leq\frac12}$,
we define its conformal weight as
\begin{equation}\label{degcw}
\Delta(g)=\Delta(a_1)+\dots+\Delta(a_s)+m_1+\dots+m_s\,.
\end{equation}
Thus we get the conformal weight space decomposition
$$
\mc V(\mf g_{\leq\frac12})=\bigoplus_{\Delta\in\frac12\mb Z_+}\mc V(\mf g_{\leq\frac12})\{\Delta\}\,.
$$
For example $\mc V(\mf g_{\leq\frac12})\{0\}=\mb F$,
$\mc V(\mf g_{\leq\frac12})\{\frac12\}=\mf g_{\frac12}$,
and $\mc V(\mf g_{\leq\frac12})\{1\}=\mf g_{0}\oplus S^2\mf g_{\frac12}$.

\begin{theorem}[\cite{DSKV13}]\label{daniele2}
Consider the PVA $\mc W=\mc W(\mf g,f)$ with the $\lambda$-bracket $\{\cdot\,_\lambda\,\cdot\}_{z,\rho}$
defined by equation \eqref{20120511:eq3}.
\begin{enumerate}[(a)]
\item
For every element $q\in\mf g^f_{1-\Delta}$ there exists a (not necessarily unique)
element $w\in\mc W\{\Delta\}=\mc W\cap\mc V(\mf g_{\leq\frac12})\{\Delta\}$ 
of the form $w=q+g$, where 
\begin{equation}\label{20140221:eq9}
g=\sum b_1^{(m_1)}\dots b_s^{(m_s)}\in\mc V(\mf g_{\leq\frac12})\{\Delta\}\,,
\end{equation}
is a sum of products of derivatives of $\ad x$-eigenvectors 
$b_i\in\mf g_{1-\Delta_i}\subset\mf g_{\leq\frac12}$,
such that
$$
\Delta_1+\dots+\Delta_s+m_1+\dots+m_s=\Delta
\,\,\text{ and }\,\,
s+m_1+\dots+m_s>1
\,.
$$
\item
Let $\{w_j=q_j+g_j\}_{j\in J^f}$ be any collection of elements in $\mc W$ as in part (a).
(Recall, from Section \ref{slod.4}, that $\{q_j\}_{j\in J^f}$ is a basis of $\mf g^f$ consisting of
$\ad x$-eigenvectors.)
Then the differential subalgebra $\mc W\subset\mc V(\mf g_{\leq\frac12})$ 
is the algebra of differential polynomials in the variables $\{w_j\}_{j\in J^f}$.
The algebra $\mc W$ is a graded associative algebra,
graded by the conformal weights defined in \eqref{degcw}:
$\mc W
=
\mb F\oplus\mc W\{1\}\oplus\mc W\{\frac32\}\oplus\mc W\{2\}\oplus\mc W\{\frac52\}\oplus\dots$.
\end{enumerate}
\end{theorem}

\section{Generators of the \texorpdfstring{$\mc W$}{W}-algebra
\texorpdfstring{$\mc W=\mc W(\mf g,f)$}{W=(g,f)}
}\label{sec:3.2}

Recall the first of the direct sum decompositions \eqref{20140221:eq4}.
By assumption, the elements $q^0_j=q_j,\,j\in J^f$, form  a basis of $\mf g^f$,
and by construction the elements $q^n_j,\,(j,n)\in J$, with $n\geq1$,
form a basis of $[e,\mf g]$
(here we are using the notation from Section \ref{slod.4}).
Since $\mf g^f\subset\mf g_{\leq\frac12}$, we have the corresponding direct sum decomposition
\begin{equation}\label{20140221:eq7}
\mf g_{\leq\frac12}=\mf g^f\oplus[e,\mf g_{\leq-\frac12}]\,.
\end{equation}
It follows that the algebra of differential polynomials $\mc V(\mf g_{\leq\frac12})$
admits the following decomposition in a direct sum of subspaces
\begin{equation}\label{20140221:eq8}
\mc V(\mf g_{\leq\frac12})
=
\mc V(\mf g^f)
\oplus
\big\langle[e,\mf g_{\leq-\frac12}]\big\rangle_{\mc V(\mf g_{\leq\frac12})}
\,,
\end{equation}
where $\mc V(\mf g^f)$ is the algebra of differential polynomials over $\mf g^f$, 
and $\big\langle[e,\mf g_{\leq-\frac12}]\big\rangle_{\mc V(\mf g_{\leq\frac12})}$
is the differential ideal of $\mc V(\mf g_{\leq\frac12})$
generated by $[e,\mf g_{\leq-\frac12}]$.

As a consequence of Theorem \ref{daniele2}, we get the following result:
\begin{corollary}\label{20140221:cor}
For every $q\in\mf g^f$ there exists a unique
element $w=w(q)\in\mc W$ 
of the form $w=q+r$, where 
$r\in\big\langle[e,\mf g_{\leq-\frac12}]\big\rangle_{\mc V(\mf g_{\leq\frac12})}$.
Moreover, if $q\in\mf g_{1-\Delta}$,
then $w(q)$ lies in $\mc W\{\Delta\}$ and $r$ is of the form \eqref{20140221:eq9}.
Consequently, $\mc W$ coincides with the algebra of differential polynomials in the variables
$w_j=w(q_j)$, $j\in J^f$.
\end{corollary}
\begin{proof}
We prove the existence of an element $w(q)=q+r$,
with $q\in\mf g^f_{1-\Delta}$ and 
$r\in\big\langle[e,\mf g_{\leq-\frac12}]\big\rangle_{\mc V(\mf g_{\leq\frac12})}$,
by induction on $\Delta$.
For $\Delta=1$, an element $w$ given by Theorem \ref{daniele2}(a)
has the form
\begin{equation}\label{20140222:eq5}
w=q+\sum b_1b_2
\,,
\end{equation}
with $b_1,b_2\in\mf g_{\frac12}$.
Since $\mf g_{\frac12}=[e,\mf g_{-\frac12}]$,
the element $w$ is of the desired form.
For $\Delta>1$,
again by Theorem \ref{daniele2}(a)
we have an element $w\in\mc W$ of the form
$w=q+g$, with $g$ as in \eqref{20140221:eq9}.
We decompose $g$ according to the direct sum decomposition \eqref{20140221:eq8}:
$g=a+b$,
where 
\begin{equation}\label{20140221:eq10}
a=\sum c_{j_1\dots j_s}^{m_1\dots m_s}(\partial^{m_1}q_{j_1})\dots(\partial^{m_s}q_{j_s})
\in\mc V(\mf g^f)
\,,
\end{equation}
and $b\in\big\langle[e,\mf g_{\leq-\frac12}]\big\rangle_{\mc V(\mf g_{\leq\frac12})}$.
Each summand in the expression \eqref{20140221:eq10} of $a$ has conformal weight $\Delta$,
therefore each $q_j$ has conformal weight strictly smaller than $\Delta$.
Hence, by inductive assumption,
there is $w_j\in\mc W$ of the form $q_j+r_j$,
with $r_j\in\big\langle[e,\mf g_{\leq-\frac12}]\big\rangle_{\mc V(\mf g_{\leq\frac12})}$.
But then, letting 
$$
A=\sum c_{j_1\dots j_s}^{m_1\dots m_s}(\partial^{m_1}w_{j_1})\dots(\partial^{m_s}w_{j_s})
\in\mc W
\,,
$$
we have that $A-a\in\big\langle[e,\mf g_{\leq-\frac12}]\big\rangle_{\mc V(\mf g_{\leq\frac12})}$.
Therefore, $w-A$
is an element of $\mc W$ of the desired form.

Next, we claim that 
\begin{equation}\label{20140221:eq12}
\mc W\cap\big\langle[e,\mf g_{\leq-\frac12}]\big\rangle_{\mc V(\mf g_{\leq\frac12})}=0\,.
\end{equation}
Let us fix, by the existence part,
a collection of elements $w_j=q_j+r_j$, where 
$r_j\in \big\langle[e,\mf g_{\leq-\frac12}]\big\rangle_{\mc V(\mf g_{\leq\frac12})}$,
for $j\in J^f$.
By Theorem \ref{daniele2}(b),
$\mc W$ is the algebra of differential polynomials in the variables $w_j$.
Take an element
$$
\sum c_{j_1\dots j_s}^{m_1\dots m_s}(\partial^{m_1}w_{j_1})\dots(\partial^{m_s}w_{j_s})
\in\mc W\cap \big\langle[e,\mf g_{\leq-\frac12}]\big\rangle_{\mc V(\mf g_{\leq\frac12})}\,.
$$
After projecting onto $\mc V(\mf g^f)$,
according to the direct sum decomposition \eqref{20140221:eq8},
we get
$$
\sum c_{j_1\dots j_s}^{m_1\dots m_s}(\partial^{m_1}q_{j_1})\dots(\partial^{m_s}q_{j_s})
=0
\,.
$$
Hence, all the coefficients $c_{j_1\dots j_s}^{m_1\dots m_s}$ must be zero,
proving \eqref{20140221:eq12}.
The uniqueness of $w(q)=q+r\in\mc W$, with 
$r\in\big\langle[e,\mf g_{\leq-\frac12}]\big\rangle_{\mc V(\mf g_{\leq\frac12})}$,
follows immediately from \eqref{20140221:eq12}.
\end{proof}

Consider the direct sum decomposition \eqref{20140221:eq8},
and let $\pi:\,\mc V(\mf g_{\leq\frac12})\twoheadrightarrow\mc V(\mf g^f)$
be the projection onto the first summand.
According to Corollary \ref{20140221:cor},
we have an injective map $w:\,\mf g^f\hookrightarrow\mc W$,
extending to a differential algebra isomorphism
$w:\,\mc V(\mf g^f)\stackrel{\sim}{\rightarrow}\mc W$,
such that 
\begin{equation}\label{20140222:eq2}
\pi\circ w=\id_{\mc V(\mf g^f)}
\,.
\end{equation}
For $q\in\mf g^f$, we have $w(q)=q+\widetilde{r(q)}$, where
$\widetilde{r(q)}\in\big\langle[e,\mf g_{\leq-\frac12}]\big\rangle_{\mc V(\mf g_{\leq\frac12})}$.
The element $\widetilde{r(q)}$ can be expanded uniquely in the form,
\begin{equation}\label{20140222:eq1}
\widetilde{r(q)}=r(q)+r^2(q)+r^3(q)+\dots=r(q)+r^{\geq2}(q)
\,,
\end{equation}
where
$r^n(q)\in\mc V(\mf g^f)\otimes S^n\big(\mb F[\partial][e,\mf g_{\leq-\frac12}]\big)$
(due to \eqref{20140221:eq8}).

\begin{remark}\label{20140312:rem}
In \cite{DSKV14} we computed the generators of the $\mc W$-algebra $\mc W(\mf g,f)$
for the minimal and short nilpotent elements $f$,
and all the expression obtained there are of the form described above.
\end{remark}

The following theorem gives us an explicit formula for the first contribution $r(q)$ 
in the expansion \eqref{20140222:eq1}.
\begin{theorem}\label{20140221:thm1}
For $q\in\mf g^f_{-k}$, the unique element $w=w(q)\in\mc W=\mc W(\mf g,f)$, 
given by Corollary \ref{20140221:cor},
has the form $w(q)=q+r(q)+r^{\geq2}(q)$ as in \eqref{20140222:eq1}, where
\begin{equation}\label{20140222:eq3}
\begin{array}{l}
\displaystyle{
\vphantom{\Big(}
r(q)=
\sum_{\frac12\leq k_1\leq k}
\sum_{(j,n)\in J_{-k_1}}
\big([q,q^{j}_{n}]^\sharp-(q|q^{j}_{n})\partial\big)q^{n+1}_{j}
} \\
\displaystyle{
\vphantom{\Big(}
+\!\!\!\!\!\!\!\!\!
\sum_{\substack{ \frac12\leq k_1\leq k \\ \frac12\leq k_2\leq k_1-1} }
\sum_{\substack{ (j_1,n_1)\in J_{-k_1} \\ (j_2,n_2)\in J_{-k_2}} }
\!\!\!\!\!\!\!\!\!
\big([q,q^{j_1}_{n_1}]^\sharp-(q|q^{j_1}_{n_1})\partial\big)
\big([q^{n_1+1}_{j_1},q^{j_2}_{n_2}]^\sharp-(q^{n_1+1}_{j_1}|q^{j_2}_{n_2})\partial\big)
q^{n_2+1}_{j_2}
} \\
\displaystyle{
\vphantom{\Big(}
+\dots
\,.}
\end{array}
\end{equation}
\end{theorem}
\begin{lemma}\label{20140221:lem2}
For $h,k\in\frac12\mb Z$, $(i,m)\in J_{-h}$ and $(j,n)\in J_{-k}$,
we have
\begin{equation}\label{20140224:eq3}
\rho\{{q^i_m} {}_\lambda q^n_j\}_z
=
\left\{\begin{array}{ll}
0 & \text{ if } h-k>1 \,, \\
\delta_{i,j}\delta_{n,m+1} & \text{ if } h-k=1 \,, \\
{[q^i_m,q^n_j]}+\delta_{i,j}\delta_{m,n}\lambda+z(s|[q^i_m,q^n_j]) & \text{ if }  h-k\leq\frac12
\,.
\end{array}\right.
\end{equation}
\end{lemma}
\begin{proof}
It follows immediately from the definitions.
\end{proof}
\begin{lemma}\label{20140224:lem1}
If $r\in\mc V(\mf g^f)\mb F[\partial][e,\mf g_{\leq-\frac12}]\,\subset\mc V(\mf g_{\leq\frac12})$ 
is such that
\begin{equation}\label{20140224:eq1}
\pi\rho\{a_\lambda r\}_z=0
\,\,\text{ for all }\,\,
a\in\mf g_{\geq\frac12}
\,,
\end{equation}
then $r=0$.
\end{lemma}
\begin{proof}
Recall that, for $k\in\frac12\mb Z$, $\frac12\leq k\leq d$, a basis of $[e,\mf g_{-k}]$ is
$\{q^{n+1}_j\}$, where $(j,n)\in J_{-k}$.
Hence,
any element $r\in\mc V(\mf g^f)\mb F[\partial][e,\mf g_{\leq-\frac12}]$
has the form
\begin{equation}\label{20140224:eq2}
r=\sum_{\frac12\leq k\leq d} \sum_{(j,n)\in J_{-k}} \sum_{p=0}^N
A_{p,j,n}\partial^p q^{n+1}_j
\,\,\text{ where }\,\,
A_{p,j,n}\in\mc V(\mf g^f)
\,.
\end{equation}
Suppose, by contradiction, that $r\neq0$,
and let $K\geq\frac12$ 
be the largest value of $k$ with non-zero contribution in the sum \eqref{20140224:eq2}.
For $(i,m)\in J_{-K}$ and $(j,n)\in J_{-k}$ (so that $(j,n+1)\in J_{-k+1}$) we have, 
by equation \eqref{20140224:eq3},
$$
\rho\{{q^i_m} {}_\lambda q^{n+1}_j\}_z
=
\left\{\begin{array}{ll}
\delta_{i,j}\delta_{n,m} & \text{ if } k=K \,, \\
0 & \text{ if } k<K \\
\end{array}\right.
\,.
$$
Taking $a=q^i_m\in\mf g_K$ in equation \eqref{20140224:eq1}, we thus get
$$
0=\pi\rho\{{q^i_m}_\lambda r\}_z
=
\sum_{\frac12\leq k\leq d} \sum_{(j,n)\in J_{-k}} \sum_{p=0}^N
A_{p,j,n} (\lambda+\partial)^p
\pi\rho\{{q^i_m}_\lambda q^{n+1}_j \}_z
=
\sum_{p=0}^N
A_{p,i,m}\lambda^p
\,.
$$
Hence, $A_{p,i,m}=0$ for all $(i,m)\in J_{-K}$, contradicting the assumption that $r\neq0$.
\end{proof}
\begin{proof}[Proof of Theorem \ref{20140221:thm1}]
By the definition \eqref{20120511:eq2} of the $\mc W$-algebra,
we have
$$
\rho\{a_\lambda w\}_z=
\rho\{a_\lambda q+r(q)+r^{\geq2}(q)\}_z
=
0
\,\,\text{ for all }\,\,
a\in\mf g_{\geq\frac12}
\,.
$$
On the other hand, by the Leibniz rule and the definition of the projection map 
$\pi:\,\mc V(\mf g_{\leq\frac12})\twoheadrightarrow\mc V(\mf g^f)$,
we clearly have $\pi\rho\{a_\lambda r^{\geq2}(q)\}_z=0$.
Therefore, thanks to Lemma \ref{20140224:lem1}, it suffices to prove that 
the element $r(q)\in\mc V(\mf g^f)\mb F[\partial][e,\mf g_{\leq-\frac12}]$ 
given by \eqref{20140222:eq3} satisfies the equation
\begin{equation}\label{20140224:eq4}
\pi\rho\{a_\lambda q+r(q)\}_z=0
\,\,\text{ for all }\,\,
a\in\mf g_{\geq\frac12}
\,.
\end{equation}
Let $h\geq\frac12$ and $(i,m)\in J_{-h}$.
By equation \eqref{20140224:eq3}, we have, for $q\in\mf g^f_{-k}$,
\begin{equation}\label{20140224:eq5a}
\pi\rho\{{q^i_m} {}_\lambda q\}_z
=
{[q^i_m,q]^\sharp}+(q^i_m|q)\lambda 
\,.
\end{equation}
Furthermore, by equation \eqref{20140224:eq3},
for $k\geq\frac12$ and $(j,n)\in J_{-k}$, we have
\begin{equation}\label{20140224:eq5b}
\pi\rho\{{q^i_m} {}_\lambda q^{n+1}_j\}_z
=
\left\{\begin{array}{ll}
0 & \text{ if } h>k \text{ or } h=k-\frac12 \,, \\
\delta_{i,j}\delta_{n,m} & \text{ if } h=k \,, \\
{[q^i_m,q^{n+1}_j]^\sharp}+(q^i_m|q^{n+1}_j)\lambda & \text{ if }  h\leq k-1
\,.
\end{array}\right.
\end{equation}
Recalling the definition \eqref{20140222:eq3} of $r(q)$, we have
\begin{equation}\label{20140224:eq6}
\begin{array}{l}
\displaystyle{
\pi\rho\{{q^i_m} {}_\lambda r(q)\}_z
=
\sum_{\frac12\leq k_1\leq k}
\sum_{(j,n)\in J_{-k_1}}
\big([q,q^{j}_{n}]^\sharp-(q|q^{j}_{n})(\partial+\lambda)\big)
\pi\rho\{{q^i_m} {}_\lambda q^{n+1}_{j}
\}_z
} \\
\displaystyle{
+
\sum_{\substack{ \frac12\leq k_1\leq k \\ \frac12\leq k_2\leq k_1-1} }
\sum_{\substack{ (j_1,n_1)\in J_{-k_1} \\ (j_2,n_2)\in J_{-k_2}} }
\big([q,q^{j_1}_{n_1}]^\sharp-(q|q^{j_1}_{n_1})(\partial+\lambda)\big)
\times} \\
\displaystyle{
\vphantom{\bigg(}
\,\,\,\,\,\,\,\,\, \,\,\,\,\,\,\,\,\, \,\,\,\,\,\,\,\,\, \,\,\,\,\,\,\,\,\,
\times
\big([q^{n_1+1}_{j_1},q^{j_2}_{n_2}]^\sharp-(q^{n_1+1}_{j_1}|q^{j_2}_{n_2})(\partial+\lambda)\big)
\pi\rho\{{q^i_m} {}_\lambda 
q^{n_2+1}_{j_2}
\}_z
+\dots
\,.
}
\end{array}
\end{equation}
By \eqref{20140224:eq5b}, the first sum in the RHS of \eqref{20140224:eq6} is equal to
\begin{equation}\label{20140224:eq7}
\begin{array}{l}
\displaystyle{
\vphantom{\Big(}
\big([q,q^{i}_{m}]^\sharp-(q|q^{i}_{m})\lambda\big)
} \\
\displaystyle{
+
\sum_{\substack{ \frac12\leq k_1\leq k \\ (h\leq k_1-1) }}
\sum_{(j,n)\in J_{-k_1}}
\big([q,q^{j}_{n}]^\sharp-(q|q^{j}_{n})(\partial+\lambda)\big)
\big({[q^i_m,q^{n+1}_j]^\sharp}+(q^i_m|q^{n+1}_j)\lambda\big)
\,.}
\end{array}
\end{equation}
Similarly, the second sum in the RHS of \eqref{20140224:eq6} is equal to
\begin{equation}\label{20140224:eq8}
\begin{array}{l}
\displaystyle{
\sum_{\substack{ \frac12\leq k_1\leq k \\ (h\leq k_1-1)} }
\sum_{(j_1,n_1)\in J_{-k_1}}
\big([q,q^{j_1}_{n_1}]^\sharp-(q|q^{j_1}_{n_1})(\partial+\lambda)\big)
\big([q^{n_1+1}_{j_1},q^{i}_{m}]^\sharp-(q^{n_1+1}_{j_1}|q^{i}_{m})\lambda\big)
} \\
\displaystyle{
+
\sum_{\substack{ \frac12\leq k_1\leq k \\ \frac12\leq k_2\leq k_1-1 \\ (h\leq k_2-1) }}
\sum_{\substack{ (j_1,n_1)\in J_{-k_1} \\ (j_2,n_2)\in J_{-k_2}} }
\big([q,q^{j_1}_{n_1}]^\sharp-(q|q^{j_1}_{n_1})(\partial+\lambda)\big)
\times} \\
\displaystyle{
\,\,\,\,\,\,\,\,\, \,\,\,\,\,\,\,\,\, \,\,\,\,\,\,\,\,\,
\times
\big([q^{n_1+1}_{j_1},q^{j_2}_{n_2}]^\sharp-(q^{n_1+1}_{j_1}|q^{j_2}_{n_2})(\partial+\lambda)\big)
\big({[q^i_m,q^{n_2+1}_{j_2}]^\sharp}+(q^i_m|q^{n_2+1}_{j_2})\lambda\big)
\,.
}
\end{array}
\end{equation}
The RHS of \eqref{20140224:eq5a} is opposite to the first term in \eqref{20140224:eq7}.
The second term in \eqref{20140224:eq7} is opposite to the first sum in \eqref{20140224:eq8}.
Proceeding by induction, we conclude that \eqref{20140224:eq4} holds.
\end{proof}

We can rewrite equation \eqref{20140222:eq3} in a more compact form as follows.
For $h,k\in\frac12\mb Z$, we introduce the notation
\begin{equation}\label{20140304:eq5}
h\prec k
\,\,\text{ if and only if }\,\,
h\leq k-1
\,.
\end{equation}
Also, for $s\geq1$, we denote $\vec{k}=(k_1,k_2,\dots,k_s)\in(\frac12\mb Z)^s$,
and $J_{-\vec{k}}:=J_{-k_1}\times\dots J_{-k_s}$.
Therefore, 
an element $(\vec{j},\vec{n})\in J_{-\vec{k}}$ is an $s$-tuple with
\begin{equation}\label{20140304:eq6}
(j_1,n_1)\in J_{-k_1},\dots,(j_s,n_s)\in J_{-k_s}
\,.
\end{equation}
Using this notation, 
equation \eqref{20140222:eq3} can be equivalently rewritten as
\begin{equation}\label{20140222:eq3b}
\begin{array}{l}
\displaystyle{
\vphantom{\Big(}
w(q_{j_0})
=
q_{j_0}
+
\sum_{\frac12\leq k_1\leq k} 
\sum_{(j_1,n_1)\in J_{-k_1}}
\big([q_{j_0},q^{j_1}_{n_1}]^\sharp
-(q_{j_0}|q^{j_1}_{n_1})\partial\big)
q^{n_1+1}_{j_1}
} \\
\displaystyle{
\vphantom{\Big(}
+
\!
\sum_{\frac12\leq k_2\prec k_1\leq k} 
\sum_{(\vec{j},\vec{n})\in J_{-\vec{k}}}
\!
\big([q_{j_0},q^{j_1}_{n_1}]^\sharp
-(q_{j_0}|q^{j_1}_{n_1})\partial\big)
\big([q^{n_1+1}_{j_1},q^{j_2}_{n_2}]^\sharp
-(q^{n_1+1}_{j_1}|q^{j_2}_{n_2})\partial\big)
q^{n_2+1}_{j_2}
} \\
\displaystyle{
\vphantom{\Big(}
+\dots+r^{\geq2}(q_{j_0})
} \\
\displaystyle{
\vphantom{\Big(}
=
\sum_{s=0}^\infty
\sum_{
\substack{ \frac12\leq k_s\prec\dots\prec k_1\leq k \\ (\vec{j},\vec{n})\in J_{-\vec{k}} }
}
\prod_{\alpha=0}^{s-1}
\big([q^{n_\alpha+1}_{j_\alpha},q^{j_{\alpha+1}}_{n_{\alpha+1}}]^\sharp
-(q^{n_\alpha+1}_{j_\alpha}|q^{j_{\alpha+1}}_{n_{\alpha+1}})\partial\big)
\cdot
q^{n_s+1}_{j_s}
+r^{\geq2}(q_{j_0})
\,,}
\end{array}
\end{equation}
where, in the RHS of the last equation, 
we are letting $n_0=-1$ and the product is taken in the increasing order of $\alpha$
(the factors do not commute).

Now we rewrite equation \eqref{20140222:eq3b} in some special cases.
For $j_0\in J^f_0$ (i.e. $\Delta=1$), 
we have a non-zero contribution in the RHS of \eqref{20140222:eq3b} only for $s=0$.
Hence, equation \eqref{20140222:eq3b} gives, in this case,
$$
w(q_{j_0})=q_{j_0}+r^{\geq2}(q_{j_0})
\,,
$$
in agreement with equation \eqref{20140222:eq5}.
For $j_0\in J^f_{-\frac12}$ (i.e. $\Delta=\frac32$), 
we have a non-zero contribution in the RHS of \eqref{20140222:eq3b} only for $s=0$ or $1$,
and for $s=1$ we must have $k_1=\frac12$.
Hence, equation \eqref{20140222:eq3b} gives, in this case,
$$
w(q_{j_0})=
q_{j_0}
-\partial q^{1}_{j_0}
+\sum_{(j_1,n_1)\in J_{-\frac12}}
[q_{j_0},q^{j_1}_{n_1}]^\sharp q^{n_1+1}_{j_1}
+r^{\geq2}(q_{j_0})
\,.
$$
For $j_0\in J^f_{-1}$ (i.e. $\Delta=2$), 
again we have a non-zero contribution in the RHS of \eqref{20140222:eq3b} 
only for $s=0$ or $1$,
and for $s=1$ we must have $k_1=\frac12$ or $1$.
Hence, equation \eqref{20140222:eq3b} gives
$$
w(q_{j_0})=
q_{j_0}
-\partial q^{1}_{j_0}
+\sum_{(j_1,n_1)\in J_{-\frac12}\sqcup J_{-1}} [q_{j_0},q^{j_1}_{n_1}]^\sharp q^{n_1+1}_{j_1}
+
r^{\geq2}(q_{j_0})
\,.
$$
We can write, for arbitrary $j_0\in J^f_{-k}$ (corresponding to conformal weight $\Delta=k+1$),
the first few summands of \eqref{20140222:eq3b} more explicitly.
The term corresponding to $s=0$ is $q_{j_0}$.
The term corresponding to $s=1$ in \eqref{20140222:eq3b} is 
$$
\sum_{\frac12\leq k_1\leq k}\sum_{(j_1,n_1)\in J_{-k_1}} 
[q_{j_0},q^{j_1}_{n_1}]^\sharp q^{n_1+1}_{j_1}
-\partial q^{1}_{j_0}
\,.
$$
Moreover, the term corresponding to $s=2$ in \eqref{20140222:eq3b} is equal to
$$
\begin{array}{l}
\displaystyle{
\sum_{\frac12\leq k_2\prec k_1\leq k}
\sum_{\substack{ (j_1,n_1)\in J_{-k_1} \\ (j_2,n_2)\in J_{-k_2} }} 
[q_{j_0},q^{j_1}_{n_1}]^\sharp [q^{n_1+1}_{j_1},q^{j_2}_{n_2}]^\sharp q^{n_2+1}_{j_2}
+\partial^2 q^{2}_{j_0}
} \\
\displaystyle{
- \sum_{\frac12\leq k_2\leq k-1} \sum_{(j_2,n_2)\in J_{-k_2}} 
\partial\big( [q^{1}_{j_0},q^{j_2}_{n_2}]^\sharp q^{n_2+1}_{j_2} \big)
- \sum_{\frac12\leq k_1\leq k} \sum_{(j_1,n_1)\in J_{-k_1}}
[q_{j_0},q^{j_1}_{n_1}]^\sharp \partial q^{n_1+2}_{j_1}
\,.}
\end{array}
$$

\section{Explicit formula for the PVA structure of
\texorpdfstring{$\mc W=\mc W(\mf g,f)$}{W=W(g,f)}}\label{sec:3.3}

\begin{proposition}\label{20140224:thm}
For $i_0\in J^f_{-h}$ and $j_0\in J^f_{-k}$, we have
\begin{equation}\label{20140224:eq9}
\begin{array}{l}
\displaystyle{
\phantom{\Big(}
\{{w(q_{i_0})}_\lambda{w(q_{j_0})}\}_{z,\rho}
=
\sum_{s,t=0}^\infty
\sum_{
\substack{ \frac12\leq h_s\prec\dots\prec h_1\leq h \\ \frac12\leq k_t\prec\dots\prec k_1\leq k }
}
\sum_{
\substack{ (\vec{i},\vec{m})\in J_{-\vec{h}} \\ (\vec{j},\vec{n})\in J_{-\vec{k}} }
}
(-1)^s
} \\
\displaystyle{
\phantom{\Big(}
\prod_{\beta=0}^{t-1}
\Big(w([q^{n_\beta+1}_{j_\beta},q^{j_{\beta+1}}_{n_{\beta+1}}]^\sharp)
-(q^{n_\beta+1}_{j_\beta}|q^{j_{\beta+1}}_{n_{\beta+1}}) (\lambda+\partial)\Big)
} \\
\displaystyle{
\phantom{\Big(}
\Big(
w([q^{m_s+1}_{i_s},q^{n_t+1}_{j_t}]^\sharp)
+(f|[q^{m_s+1}_{i_s},q^{n_t+1}_{j_t}])
} \\
\displaystyle{
\phantom{\Big(}
\,\,\,\,\,\,\,\,\, \,\,\,\,\,\,\,\,\,
+(q^{m_s+1}_{i_s}|q^{n_t+1}_{j_t})(\lambda+\partial)
+z(s|[q^{m_s+1}_{i_s},q^{n_t+1}_{j_t}])
\Big)
} \\
\displaystyle{
\phantom{\Big(}
\prod_{\alpha=0}^{s-1}
\Big(
w([q^{i_{s-\alpha}}_{m_{s-\alpha}},q^{m_{s-\alpha-1}+1}_{i_{s-\alpha-1}}]^\sharp)
-(q^{i_{s-\alpha}}_{m_{s-\alpha}}|q^{m_{s-\alpha-1}+1}_{i_{s-\alpha-1}}) (\lambda+\partial)
\Big)
\,,}
\end{array}
\end{equation}
where the products are taken in the increasing order of the indices $\alpha$ and $\beta$.
\end{proposition}
\begin{proof}
According to Corollary \ref{20140221:cor},
the maps $w:\,\mc V(\mf g^f)$ and $\pi|_{\mc W}:\,\mc W\to\mc V(\mf g^f)$
are inverse to each other.
Hence, in order to compute $\{{w(q_{i_0})}_\lambda{w(q_{j_0})}\}_{z,\rho}\in\mc W[\lambda]$,
we can first compute
$$
\pi\rho\{{w(q_{i_0})}_\lambda{w(q_{j_0})}\}_{z}\,\in\mc V(\mf g^f)\,,
$$
and then apply the differential algebra isomorphism  
$w:\,\mc V(\mf g^f)\to\mc W$ to the result.
On the other hand, it is clear from the Leibniz rules and the definition of the map 
$\pi:\,\mc V(\mf g_{\leq\frac12})\twoheadrightarrow\mc V(\mf g^f)$,
that, in the expansions 
$w(q_{i_0})=q_{i_0}+r(q_{i_0})+r^{\geq2}(q_{i_0})$
and $w(q_{j_0})=q_{j_0}+r(q_{j_0})+r^{\geq2}(q_{j_0})$
(cf. equation \eqref{20140222:eq1}),
the terms $r^{\geq2}(q_{i_0})$ and $r^{\geq2}(q_{j_0})$ give no contribution.
Moreover, by equation \eqref{20140222:eq3b}
we get, using the Leibniz rules and the definition of the map $\pi$,
$$
\begin{array}{l}
\displaystyle{
\phantom{\Big(}
\pi\{{w(q_{i_0})}_\lambda{w(q_{j_0})}\}_{z,\rho}
} \\
\displaystyle{
\phantom{\Big(}
=
\sum_{s,t=0}^\infty
\sum_{
\substack{ \frac12\leq h_s\prec\dots\prec h_1\leq h \\ \frac12\leq k_t\prec\dots\prec k_1\leq k }
}
\sum_{
\substack{ (\vec{i},\vec{m})\in J_{-\vec{h}} \\ (\vec{j},\vec{n})\in J_{-\vec{k}} }
}
\prod_{\beta=0}^{t-1}
\Big([q^{n_\beta+1}_{j_\beta},q^{j_{\beta+1}}_{n_{\beta+1}}]^\sharp
-(q^{n_\beta+1}_{j_\beta}|q^{j_{\beta+1}}_{n_{\beta+1}})(\lambda+\partial)\Big)
} \\
\displaystyle{
\phantom{\Big(}
\times
{\pi\rho\{ {q^{m_s+1}_{i_s}} {}_{\lambda+\partial}\, {q^{n_t+1}_{j_t}} \}_z}_{\to}
\prod_{\alpha=s-1}^{0}
\Big([q^{m_\alpha+1}_{i_\alpha},q^{i_{\alpha+1}}_{m_{\alpha+1}}]^\sharp
-(q^{m_\alpha+1}_{i_\alpha}|q^{i_{\alpha+1}}_{m_{\alpha+1}})(-\lambda-\partial)\Big)
\,,}
\end{array}
$$
where the second product is taken in the decreasing order of $\alpha$.
Equation \eqref{20140224:eq9} is the result of applying the map $w$ to both sides of this equation.
\end{proof}

Formula \eqref{20140224:eq9} has the advantage of being manifestly skewsymmetric.
It can be simplified by bringing it to the form similar to \eqref{20140319:eq7},
at the price of loosing the manifest skewsymmetry.
For this, we will need the following result.
\begin{lemma}\label{20140304:lem}
For every $k\in\frac12\mb Z$, $-d\leq k\leq d$, we have
\begin{equation}\label{20140304:eq1}
\sum_{(i,m)\in J_{k-1}} q^{m+1}_i\otimes q^i_m
=-\sum_{(j,n)\in J_{-k}}q^j_n\otimes q^{n+1}_j
\,\,\in[e,\mf g_{k-1}]\otimes[e,\mf g_{-k}]
\,.
\end{equation}
\end{lemma}
\begin{proof}
First, we prove that both the LHS and the RHS of \eqref{20140304:eq1} 
lie in $[e,\mf g_{k-1}]\otimes[e,\mf g_{-k}]$.
By assumption $\{q^j_n\}_{(j,n)\in J_{-k}}$ is a basis of $\mf g_k=[e,\mf g_{k-1}]\oplus\mf g^f_k$.
On the other hand, for $q^j_n\in\mf g^f_k$,
we have $q^n_j\in\mf g^e_{-k}$,
and therefore $q^{n+1}_j=$const.$[e,q^n_j]=0$.
Hence, only the basis elements $q^j_n$ in $[e,\mf g_{k-1}]$ give a non-zero contribution 
to $\sum_{(j,n)\in J_{-k}}q^j_n\otimes q^{n+1}_j$,
which therefore lies in $[e,\mf g_{k-1}]\otimes[e,\mf g_{-k}]$.
The same argument applies to the LHS of \eqref{20140304:eq1}.

Note that 
the space $[e,\mf g_{k-1}]$ is non-degenerately paired by the bilinear form $(\cdot\,|\,\cdot)$
to $[f,\mf g_{-k+1}]$,
and the space $[e,\mf g_{-k}]$ is non-degenerately paired 
to $[f,\mf g_{k}]$.
In fact, 
the elements $\{q^{m+1}_i\}_{(i,m)\in J_{k-1}}$
form a basis of $[e,\mf g_{k-1}]$,
and $\{q_{m+1}^i\}_{(i,m)\in J_{k-1}}$
is the dual basis of $[f,\mf g_{-k+1}]$.
Similarly, $\{q^{n+1}_j\}_{(j,n)\in J_{-k}}$
form a basis of $[e,\mf g_{-k}]$,
and $\{q_{n+1}^j\}_{(j,n)\in J_{-k}}$
is the dual basis of $[f,\mf g_{k}]$.
Hence, we have the following completeness relations
\begin{equation}\label{20140304:eq3}
\begin{array}{l}
\displaystyle{
\vphantom{\Big(}
\sum_{(i,m)\in J_{k-1}} 
([f,a]|q^{m+1}_i)q^i_{m+1}
=[f,a]
\,\,\text{ for every }\,\,
a\in\mf g_{-k+1}
\,, } \\
\displaystyle{
\vphantom{\Big(}
\sum_{(j,n)\in J_{-k}} 
([f,b]|q^{n+1}_j)q^j_{n+1}
=[f,b]
\,\,\text{ for every }\,\,
b\in\mf g_k
\,.}
\end{array}
\end{equation}

Since both the LHS and the RHS of \eqref{20140304:eq1} 
lie in $[e,\mf g_{k-1}]\otimes[e,\mf g_{-k}]$,
to prove the equality in \eqref{20140304:eq1}
it suffices to show that, for every $a\in\mf g_{-k+1}$ and $b\in\mf g_k$,
we have
\begin{equation}\label{20140304:eq2}
\sum_{(i,m)\in J_{k-1}} 
([f,a]|q^{m+1}_i)([f,b]|q^i_m)
=
-
\sum_{(j,n)\in J_{-k}}
([f,a]|q^j_n)([f,b]|q^{n+1}_j)
\,.
\end{equation}
But by the invariance of the bilinear form and the completeness relations \eqref{20140304:eq3},
both sides of \eqref{20140304:eq2} are equal to $([a,f]|b)$.
\end{proof}
\begin{theorem}\label{20140304:thm}
For $a\in \mf g^f_{-h}$ and $b\in\mf g^f_{-k}$, we have
\begin{equation}\label{20140304:eq4}
\begin{array}{l}
\displaystyle{
\phantom{\Big(}
\{w(a)_\lambda w(b)\}_{z,\rho}
=
w([a,b])+(a|b)\lambda+z(s|[a,b])
} \\
\displaystyle{
\phantom{\Big(}
-\sum_{t=1}^\infty
\sum_{
-h+1\leq k_t\prec \dots\prec k_1\leq k
}
\sum_{
(\vec{j},\vec{n})\in J_{-\vec{k}}
}
\big(
w([b,q^{j_1}_{n_1}]^\sharp)
-(b|q^{j_1}_{n_1})(\lambda+\partial)
+z(s|[b,q^{j_{1}}_{n_{1}}])
\big)
} \\
\displaystyle{
\phantom{\Big(}
\big(
w([q^{n_1+1}_{j_1},q^{j_2}_{n_2}]^\sharp)
-(q^{n_1+1}_{j_1}|q^{j_2}_{n_2}) (\lambda+\partial)
+z(s|[q^{n_1+1}_{j_1},q^{j_2}_{n_2}])
\big)
\dots } \\
\displaystyle{
\phantom{\Big(}
\dots
\big(
w([q^{n_{t-1}+1}_{j_{t-1}},q^{j_t}_{n_t}]^\sharp)
-(q^{n_{t-1}+1}_{j_{t-1}}|q^{j_t}_{n_t}) (\lambda+\partial)
+z(s|[q^{n_{t-1}+1}_{j_{t-1}},q^{j_t}_{n_t}])
\big)
} \\
\displaystyle{
\phantom{\Big(}
\big(
w([q^{n_t+1}_{j_t},a]^\sharp)
-(q^{n_t+1}_{j_t}|a) \lambda
+z(s|[q^{n_t+1}_{j_t},a])
\big)
\,.}
\end{array}
\end{equation}
\end{theorem}
Note that in each summand of \eqref{20140304:eq4} 
the $z$ term can be non-zero at most in one factor.
In fact, $z$ may occur in the first factor only for $k_1\leq0$,
in the second factor only for $k_1\geq1$ and $k_2\leq-1$,
in the third factor only for $k_2\geq1$ and $k_3\leq-1$,
and so on, and it may occur in the last factor only for $k_t\geq1$.
Since these conditions are mutually exclusive,
the expression in the RHS of \eqref{20140304:eq4} is linear in $z$.
\begin{proof}
By Lemma \ref{20140304:lem}, we have, for $\alpha=0,\dots,s-1$,
$$
\sum_{(i_{s-\alpha},m_{s-\alpha})\in J_{-h_{s-\alpha}}}
\!
q^{m_{s-\alpha}+1}_{i_{s-\alpha}}\otimes q_{m_{s-\alpha}}^{i_{s-\alpha}}
=
-\sum_{(j_{t+1+\alpha},n_{t+1+\alpha})\in J_{-k_{t+1+\alpha}}}
\!
q_{n_{t+1+\alpha}}^{j_{t+1+\alpha}}\otimes q^{n_{t+1+\alpha}+1}_{j_{t+1+\alpha}}
\,,
$$
where
$$
k_{t+1+\alpha}=-h_{s-\alpha}+1
\,.
$$
Also, the inequalities $\frac12\leq h_s\prec\dots\prec h_1\leq h$ translate,
in terms of the new indices $k_{t+1},\dots,k_{t+s}$, to
$$
-h+1\leq k_{t+s}\prec k_{t+s-1}\prec\dots\prec k_{t+2}\prec k_{t+1}\leq\frac12
\,.
$$
Hence, formula \eqref{20140224:eq9} can be rewritten as follows
\begin{equation}\label{20140305:eq1}
\begin{array}{l}
\displaystyle{
\phantom{\Big(}
\{{w(q_{i_0})}_\lambda{w(q_{j_0})}\}_{z,\rho}
=
-\sum_{s,t=0}^\infty
\sum_{
\substack{ \frac12\leq k_t\prec\dots\prec k_1\leq k \\ -h+1\leq k_{t+s}\prec\dots\prec k_{t+1}\leq\frac12
}
}
\sum_{
(\vec{j},\vec{n})\in J_{-\vec{k}}
}
} \\
\displaystyle{
\phantom{\Big(}
\prod_{\beta=0}^{t-1}
\Big(w([q^{n_\beta+1}_{j_\beta},q^{j_{\beta+1}}_{n_{\beta+1}}]^\sharp)
-(q^{n_\beta+1}_{j_\beta}|q^{j_{\beta+1}}_{n_{\beta+1}}) (\lambda+\partial)\Big)
} \\
\displaystyle{
\phantom{\Big(}
\Big(
w([q^{n_t+1}_{j_t},q_{n_{t+1}}^{j_{t+1}}]^\sharp)
-(f|[q_{n_{t+1}}^{j_{t+1}},q^{n_t+1}_{j_t}])
} \\
\displaystyle{
\phantom{\Big(}
\,\,\,\,\,\,\,\,\, \,\,\,\,\,\,\,\,\,
-(q^{n_t+1}_{j_t}|q_{n_{t+1}}^{j_{t+1}})(\lambda+\partial)
+z(s|[q^{n_t+1}_{j_t},q_{n_{t+1}}^{j_{t+1}}])
\Big)
} \\
\displaystyle{
\phantom{\Big(}
\prod_{\alpha=0}^{s-1}
\Big(
w([q^{n_{t+1+\alpha}+1}_{j_{t+1+\alpha}},q_{n_{t+2+\alpha}}^{j_{t+2+\alpha}}]^\sharp)
-(q^{n_{t+1+\alpha}+1}_{j_{t+1+\alpha}}|q_{n_{t+2+\alpha}}^{j_{t+2+\alpha}}) (\lambda+\partial)
\Big)
\,,}
\end{array}
\end{equation}
where we let $q_{n_{t+s+1}}^{j_{t+s+1}}=q_{i_0}$ in the last factor.
First,
recall that $q_{n_{t+1}}^{j_{t+1}}\in\mf g_{k_{t+1}}$
and $q^{n_t+1}_{j_t}\in[e,\mf g_{-k_t}]\subset\mf g_{-k_t+1}$.
Since, by assumption, $s\in\mf g_d$,
it follows that $(s|[q^{n_t+1}_{j_t},q_{n_{t+1}}^{j_{t+1}}])$
can be non-zero only if $k_t\geq1$ and $k_{t+1}\leq-1$.
Therefore, the coefficient of $z$ in formula \eqref{20140305:eq1}
is the same as the coefficient of $z$ in formula \eqref{20140304:eq4},
(for $a=q_{i_0}$ and $b=q_{j_0}$).

Next, we study formula \eqref{20140305:eq1} when $z=0$.
We consider separately the two contributions 
to the RHS of \eqref{20140305:eq1}:
the one in which the term $(f|[q_{n_{t+1}}^{j_{t+1}},q^{n_t+1}_{j_t}])$ enters,
and the remainder.
For the first contribution, we note that
$$
(f|[q_{n_{t+1}}^{j_{t+1}},q^{n_t+1}_{j_t}])
=
(q_{n_{t+1}+1}^{j_{t+1}}|q^{n_t+1}_{j_t})
=
\delta_{j_t,j_{t+1}}\delta_{n_t,n_{t+1}}
\,.
$$
Therefore, 
the term in the RHS of \eqref{20140305:eq1}
in which $(f|[q_{n_{t+1}}^{j_{t+1}},q^{n_t+1}_{j_t}])$ enters is
(letting $\ell=t+s-1$ and $q_{n_\ell}^{j_\ell}=q_{i_0}$)
\begin{equation}\label{20140305:eq2}
\begin{array}{l}
\displaystyle{
\phantom{\Big(}
+\sum_{\ell=0}^\infty
\sum_{t=0}^\ell
\sum_{
\substack{ -h+1\leq k_\ell\prec\dots\prec k_1\leq k \\ k_t=\frac12
}
}
\sum_{
(\vec{j},\vec{n})\in J_{-\vec{k}}
}
} \\
\displaystyle{
\phantom{\Big(}
\prod_{\beta=0}^{\ell-1}
\Big(w([q^{n_\beta+1}_{j_\beta},q^{j_{\beta+1}}_{n_{\beta+1}}]^\sharp)
-(q^{n_\beta+1}_{j_\beta}|q^{j_{\beta+1}}_{n_{\beta+1}}) (\lambda+\partial)\Big)
\,.}
\end{array}
\end{equation}
Next, note that, if $k_t\geq\frac12$ and $k_{t+1}\leq\frac12$,
then either $k_{t+1}\prec k_t$,
or $(k_t,k_{t+1})$ is one of the following three pairs:
$(\frac12,0)$, $(\frac12,\frac12)$, $(1,\frac12)$.
But in all these three cases, we have
$[q^{n_t+1}_{j_t},q_{n_{t+1}}^{j_{t+1}}]^\sharp=0$
and $(q^{n_t+1}_{j_t}|q_{n_{t+1}}^{j_{t+1}})=0$.
Therefore, 
the contribution to the RHS of \eqref{20140305:eq1},
in which the term $(f|[q_{n_{t+1}}^{j_{t+1}},q^{n_t+1}_{j_t}])$ 
does not enter, is (letting $\ell=t+s$)
\begin{equation}\label{20140305:eq3}
\begin{array}{l}
\displaystyle{
\phantom{\Big(}
-\sum_{\ell=0}^\infty
\sum_{t=0}^\ell
\sum_{
\substack{ -h+1\leq k_\ell\prec\dots\prec k_1\leq k \\ k_{t+1}\leq\frac12\leq k_t
}
}
\sum_{
(\vec{j},\vec{n})\in J_{-\vec{k}}
}
} \\
\displaystyle{
\phantom{\Big(}
\prod_{\beta=0}^{\ell-1}
\Big(
w([q^{n_\beta+1}_{j_\beta},q^{j_{\beta+1}}_{n_{\beta+1}}]^\sharp)
-(q^{n_\beta+1}_{j_\beta}|q^{j_{\beta+1}}_{n_{\beta+1}}) (\lambda+\partial)
\Big)
\,.} 
\end{array}
\end{equation}
The sum over the indices $\vec{k}$ in \eqref{20140305:eq3}
has terms in which $k_{t+1}<\frac12<k_t$,
terms in which $k_{t+1}<\frac12=k_t$,
and terms in which $k_{t+1}=\frac12<k_t$.
Each of the last two types of terms give the same contribution as \eqref{20140305:eq2}
but with opposite sign.
Hence, combining \eqref{20140305:eq1} and \eqref{20140305:eq2}
we get
$$
\begin{array}{l}
\displaystyle{
\phantom{\Big(}
-\sum_{\ell=0}^\infty
\sum_{
-h+1\leq k_\ell\prec\dots\prec k_1\leq k 
}
\sum_{
(\vec{j},\vec{n})\in J_{-\vec{k}}
}
} \\
\displaystyle{
\phantom{\Big(}
\prod_{\beta=0}^{\ell-1}
\Big(
w([q^{n_\beta+1}_{j_\beta},q^{j_{\beta+1}}_{n_{\beta+1}}]^\sharp)
-(q^{n_\beta+1}_{j_\beta}|q^{j_{\beta+1}}_{n_{\beta+1}}) (\lambda+\partial)
\Big)
\,,} 
\end{array}
$$
which is the same as the RHS of \eqref{20140304:eq4} for $z=0$.
\end{proof}

\begin{remark}\label{20140310:rem}
Let $\zeta\in\mf g^e$.
Consider the differential algebra automorphism of $\mc W= S(\mb F[\partial]w(\mf g^f))$
defined, on generators, by
$$
w(a)\mapsto w(a)+(\zeta|a)
\,\,,\,\,\,\,a\in\mf g^f\,.
$$
(We could let $\zeta$ be an arbitrary element of $\mf g$,
but for $\zeta\in[f,\mf g]$ this map is the identity map.)
Under this automorphism, the PVA $\lambda$-bracket $\{\,\cdot_\lambda\,\cdot\}_{z=0,\rho}$
is mapped to the following deformed $\lambda$-bracket
\begin{equation}\label{20140304:eq4b}
\begin{array}{l}
\displaystyle{
\phantom{\Big(}
\{w(a)_\lambda w(b)\}^\zeta
=
w([a,b])+(a|b)\lambda+(\zeta|[a,b])
} \\
\displaystyle{
\phantom{\Big(}
-\sum_{t=1}^\infty
\sum_{
-h+1\leq k_t\prec \dots\prec k_1\leq k
}
\sum_{
(\vec{j},\vec{n})\in J_{-\vec{k}}
}
\big(
w([b,q^{j_1}_{n_1}]^\sharp)
-(b|q^{j_1}_{n_1})(\lambda+\partial)
+(\zeta|[b,q^{j_{1}}_{n_{1}}])
\big)
} \\
\displaystyle{
\phantom{\Big(}
\big(
w([q^{n_1+1}_{j_1},q^{j_2}_{n_2}]^\sharp)
-(q^{n_1+1}_{j_1}|q^{j_2}_{n_2}) (\lambda+\partial)
+(\zeta|[q^{n_1+1}_{j_1},q^{j_2}_{n_2}])
\big)
\dots } \\
\displaystyle{
\phantom{\Big(}
\dots
\big(
w([q^{n_{t-1}+1}_{j_{t-1}},q^{j_t}_{n_t}]^\sharp)
-(q^{n_{t-1}+1}_{j_{t-1}}|q^{j_t}_{n_t}) (\lambda+\partial)
+(\zeta|[q^{n_{t-1}+1}_{j_{t-1}},q^{j_t}_{n_t}])
\big)
} \\
\displaystyle{
\phantom{\Big(}
\big(
w([q^{n_t+1}_{j_t},a]^\sharp)
-(q^{n_t+1}_{j_t}|a) \lambda
+(\zeta|[q^{n_t+1}_{j_t},a])
\big)
\,.}
\end{array}
\end{equation}
This $\lambda$-bracket with $\zeta=zs$ coincides with the $\lambda$-bracket \eqref{20140304:eq4}.
This proves, in particular, that classical $\mc W$-algebras are isomorphic for different choices 
of $z\in\mb F$.
In fact, the $\lambda$-brackets $\{\cdot\,_\lambda\,\cdot\}^\zeta$
define a family of isomorphic PVA's parametrized by $\zeta\in\mf g^e$.
However the dependence on $\zeta$ is in general non-linear.
As we pointed out after Theorem \ref{20140304:thm}, for $\zeta=zs$ and $s\in\mf g_d$
the $\lambda$-bracket \eqref{20140304:eq4b}
is linear in $z$.
Hence, in this case, we get a compatible family of PVA's parametrized by elements of $\mf g_d$.
\end{remark}

\section{Special cases}\label{sec:3.4}

\subsection{Elements of conformal weight
\texorpdfstring{$1$}{1}}\label{sec:3.4.1}

Consider the case when either $a$ or $b$ lies in $\mf g^f_0$,
which corresponds to a generator $w(a)$ or $w(b)$ of $\mc W=\mc W(\mf g,f)$ of conformal weight $\Delta=1$.
Since $\mf g^f_0=\mf g^e_0$,
if $a\in\mf g^f_0$ and $n_t\geq0$,
we have $[q^{n_t+1}_{j_t},a]^\sharp=0$
and $(q^{n_t+1}_{j_t}|a)=0$.
For the first equation we used the fact that
\begin{equation}\label{20140227:eq3}
[[e,\mf g],\mf g^e]\subset[e,\mf g]
\,.
\end{equation}
Hence, the sum in the RHS of \eqref{20140304:eq4} is zero in this case.
The case when $b\in\mf g^f_0$ can be derived by skewsymmetry,
or by the fact that, thanks to Lemma \ref{20140304:lem},
we also have
$\sum_{(j_1,n_1)\in J_{-h_1}}[b,q^{j_1}_{n_1}]^\sharp\otimes q^{n_1+1}_{j_1}=0$
and 
$\sum_{(j_1,n_1)\in J_{-h_1}}(b|q^{j_1}_{n_1})\otimes q^{n_1+1}_{j_1}=0$.

In conclusion, 
if either $a$ or $b$ lies in $\mf g^f_0$, we have
\begin{equation}\label{20140225:eq1}
\{{w(a)}_\lambda{w(b)}\}_{z,\rho}
=
w([a,b])+(a|b)\lambda+z(s|[a,b])
\,.
\end{equation}
In particular, the map $w$ restricts to an injective PVA homomorphism
$\mc V(\mf g^f_0)\hookrightarrow\mc W$.
Furthermore, \eqref{20140225:eq1} defines a representation 
of the Lie conformal algebra $\mb F[\partial]\mf g^f_0$
on $\mb F[\partial]U_k$, where $U_k=\{b+z(s|b)\,|\,b\in\mf g^f_{-k}\}$ and $k\geq\frac12$.
(Explicitly, this representation is given by
$a {}_\lambda(b+z(s|b))=[a,b]+z(s|[a,b])$ for $a\in\mf g^f_0$, $b\in\mf g^f_{-k}$,
and extended by sesquilinearity.)

\subsection{
Elements of conformal weight \texorpdfstring{$\frac32$}{3/2}}
\label{sec:3.4.2}

Consider the case when $a,b\in\mf g^f_{-\frac12}$,
corresponding to generators $w(a)$ and $w(b)$ of $\mc W$ of conformal weight $\Delta=\frac32$.
In this case the sum over the indices $\vec{k}$ is non-empty only for $t=1$,
and in this case it must be $k_1=\frac12$.
Moreover, it is easy to check,
using Lemmas \ref{20140304:lem2} and \ref{20140304:lem}, that
$$
\sum_{(j,n)\in J_{-\frac12}} (a|q^j_n)q^{n+1}_j
=\sum_{j\in J^f_{-\frac12}} (a|q^j)q^{1}_j
=-[e,a]
\,\,,\,\,\,\,
\sum_{(j,n)\in J_{-\frac12}} (a|q^{n+1}_j)q^j_n
=[e,a]
\,.
$$

In conclusion, for $a,b\in\mf g^f_{-\frac12}$ we get
\begin{equation}\label{20140226:eq8}
\begin{array}{l}
\displaystyle{
\phantom{\Big(}
\{{w(a)}_\lambda{w(b)}\}_{z,\rho}
=
w([a,b])
+(\partial+2\lambda)
w([a,[e,b]]^\sharp)
-(e|[a,b])\lambda^2
} \\
\displaystyle{
\,\,\,\,\,\,\,\,\, \,\,\,\,\,\,\,\,\, \,\,\,\,\,\,\,\,\,
+\sum_{(j,n)\in J_{-\frac12}}
w([a,q^{j}_{n}]^\sharp)
w([q^{n+1}_{j},b]^\sharp)
+z(s|[a,b])
\,.}
\end{array}
\end{equation}
Equation \eqref{20140226:eq8} is the same as \cite[Eq.(3.11)]{DSKV14}
(cf. \cite{Suh13}).

\subsection{
Generator \texorpdfstring{$w(f)$}{w(f)}}\label{sec:3.4.3}

Next, we consider the case when $a=f$ (for which $h=1$).
In this case $[f,b]=0$ and $(f|b)=0$ for every $b\in\mf g^f$.
For $t\geq1$ we have, by Lemma \ref{20140304:lem2}(i),
\begin{equation}\label{20140305:eq4}
[q^{n_t+1}_{j_t},f]^\sharp=(q^{n_t}_{j_t})^\sharp
=\delta_{k_t\geq\frac12}\delta_{n_t,0}q_{j_t}
\,.
\end{equation}
Here and further we use the standard notation $\delta_{k\geq\frac12}$,
which is $1$ for $k\geq\frac12$ and $0$ otherwise.
By Lemma \ref{20140304:lem} we also have
\begin{equation}\label{20140305:eq5}
\sum_{(j_t,n_t)\in J_{-k_t}}
(q^{n_t+1}_{j_t}|f)q^{j_t}_{n_t}
=
-\sum_{(j_t,n_t)\in J_{k_t-1}}
(q^{j_t}_{n_t}|f)q^{n_t+1}_{j_t}
=
\delta_{k_t,0}x
\end{equation}
The term with $t=1$ in the RHS of \eqref{20140304:eq4} is,
by \eqref{20140305:eq4} and \eqref{20140305:eq5},
\begin{equation}\label{20140305:eq6}
\begin{array}{l}
\displaystyle{
\phantom{\Big(}
-
\sum_{
0\leq k_1\leq k
}
\sum_{
(j_1,n_1)\in J_{-k_1}
}
\big(
w([b,q^{j_1}_{n_1}]^\sharp)-(b|q^{j_1}_{n_1})(\lambda+\partial)
+z(s|[b,q^{j_{1}}_{n_{1}}])
\big)
} \\
\displaystyle{
\phantom{\Big(}
\times
\big(
w([q^{n_1+1}_{j_1},f]^\sharp)
-(q^{n_1+1}_{j_1}|f) \lambda
+z(s|[q^{n_1+1}_{j_1},f])
\big)
} \\
\displaystyle{
\phantom{\Big(}
=
\sum_{
\frac12\leq k_1\leq k
}
\sum_{
j_1\in J^f_{-k_1}
}
w(q_{j_1})
w([q^{j_1},b]^\sharp)
+
\delta_{k\geq\frac12}
(\lambda+(k+1)\partial)w(b)
} \\
\displaystyle{
\phantom{\Big(}
-
zw([b,s]^\sharp)
+
(k+1) z (s|b) \lambda
\,.}
\end{array}
\end{equation}
For $t\geq2$, the corresponding summand in the RHS of \eqref{20140304:eq4} is,
again by \eqref{20140305:eq4} and \eqref{20140305:eq5},
\begin{equation}\label{20140304:eq7}
\begin{array}{l}
\displaystyle{
\phantom{\Big(}
-
\sum_{
0\leq k_t\prec \dots\prec k_1\leq k
}
\sum_{
(\vec{j},\vec{n})\in J_{-\vec{k}}
}
\big(
w([b,q^{j_1}_{n_1}]^\sharp)
-(b|q^{j_1}_{n_1})(\lambda+\partial)
+z(s|[b,q^{j_{1}}_{n_{1}}])
\big)
\dots } \\
\displaystyle{
\phantom{\Big(}
\dots
\big(
w([q^{n_{t-1}+1}_{j_{t-1}},q^{j_t}_{n_t}]^\sharp)
-(q^{n_{t-1}+1}_{j_{t-1}}|q^{j_t}_{n_t}) (\lambda+\partial)
+z(s|[q^{n_{t-1}+1}_{j_{t-1}},q^{j_t}_{n_t}])
\big)
} \\
\displaystyle{
\phantom{\Big(}
\big(
w([q^{n_t+1}_{j_t},f]^\sharp)
-(q^{n_t+1}_{j_t}|f) \lambda
+z(s|[q^{n_t+1}_{j_t},f])
\big)
} \\
\displaystyle{
\phantom{\Big(}
=
-\sum_{
1\leq k_{t-1}\prec \dots\prec k_1\leq k
}
\sum_{
(\vec{j},\vec{n})\in J_{-\vec{k}}
}
\big(
w([b,q^{j_1}_{n_1}]^\sharp)
-(b|q^{j_1}_{n_1})(\lambda+\partial)
+z(s|[b,q^{j_{1}}_{n_{1}}])
\big)
\dots } \\
\displaystyle{
\phantom{\Big(}
\dots
\big(
w([q^{n_{t-2}+1}_{j_{t-2}},q^{j_{t-1}}_{n_{t-1}}]^\sharp)
-(q^{n_{t-2}+1}_{j_{t-2}}|q^{j_{t-1}}_{n_{t-1}}) (\lambda+\partial)
+z(s|[q^{n_{t-2}+1}_{j_{t-2}},q^{j_{t-1}}_{n_{t-1}}])
\big)
} \\
\displaystyle{
\phantom{\Big(}
(q^{n_{t-1}+1}_{j_{t-1}}|x) \lambda^2
\,.}
\end{array}
\end{equation}
Here we used that fact that, for $n_{t-1}\geq0$, we have
$[q^{n_{t-1}+1}_{j_{t-1}},q^{j_t}]^\sharp=0$ 
(by \eqref{20140227:eq3}), $(q^{n_{t-1}+1}_{j_{t-1}}|q^{j_t})=0$,
$(s|[q^{n_{t-1}+1}_{j_{t-1}},q^{j_t}])=0$,
and $[q^{n_{t-1}+1}_{j_{t-1}},x]^\sharp=0$.
Note that $(q^{n_{t-1}+1}_{j_{t-1}}|x)$ is zero unless $n_{t-1}=0$.
But for $n_{t-1}=0$ and $n_{t-2}\geq0$, we have
$[q^{n_{t-2}+1}_{j_{t-2}},q^{j_{t-1}}]^\sharp=0$
and $(q^{n_{t-2}+1}_{j_{t-2}}|q^{j_{t-1}})=0$.
Hence, for $t\geq2$ the RHS of \eqref{20140304:eq7} vanishes.
Moreover, we have
$$
\sum_{(j_1,n_1)\in J_{-k_1}}(q^{n_1+1}_{j_1}|x)q^{j_1}_{n_1}=-\delta_{k_1,1}\frac12 e
\,.
$$
Hence, tor $t=2$ the RHS of \eqref{20140304:eq7} becomes
\begin{equation}\label{20140304:eq8}
-\frac12
(b|e)
\lambda^3
\,.
\end{equation}

Combining \eqref{20140305:eq6} and \eqref{20140304:eq8},
we conclude that, for $a\in\mf g^f_{1-\Delta}$ we have
\begin{equation}\label{20140227:eq11}
\begin{array}{l}
\displaystyle{
\phantom{\Big(}
\{{w(f)}_\lambda{w(a)}\}_{z,\rho}
=
\sum_{j\in J^f_{\leq-\frac12}}
w(q_{j})w([q^{j},a]^\sharp)
+(1-\delta_{\Delta,1})(\partial+\Delta\lambda)w(a)
} \\
\displaystyle{
\phantom{\Big(}
-\frac{(e|a)}{2}
\lambda^3
+z w([s,a]^\sharp)
+z\Delta(s|a)\lambda
\,.}
\end{array}
\end{equation}

\subsection{
Virasoro element}
\label{sec:3.4.4}

\begin{proposition}\phantomsection\label{20140228:prop}
\begin{enumerate}[(a)]
\item
Consider the element
$L_0=\frac12\sum_{j\in J^f_0}w(q_j)w(q^j)\,\in\mc W\{2\}$.
For $a\in\mf g^f_{-k}$, we have
\begin{equation}\label{20140228:eq1}
\begin{array}{l}
\displaystyle{
\vphantom{\Big(}
\{{L_0}_\lambda w(a)\}_{z,\rho}
=
\sum_{j\in J^f_0}
w(q_{j})w([q^{j},a])
+
\delta_{k,0}(\partial+\lambda)w(a)
-
\delta_{k,d} zw([s,a]^\sharp)
\,,} \\
\displaystyle{
\vphantom{\Big(}
\{{w(a)}_\lambda L_0\}_{z,\rho}
=
-\sum_{j\in J^f_0}
w(q_{j})w([q^{j},a])
+
\delta_{k,0}w(a)\lambda
+
\delta_{k,d} zw([s,a]^\sharp)
\,,}
\end{array}
\end{equation}
In particular, for $a\in\mf g^f_0$, we have
\begin{equation}\label{20140228:eq1b}
\begin{array}{l}
\displaystyle{
\vphantom{\Big(}
\{{L_0}_\lambda w(a)\}_{z,\rho}
=
(\partial+\lambda)w(a)
\,\,,\,\,\,\,\,
\{{w(a)}_\lambda L_0\}_{z,\rho}
=
w(a)\lambda
\,.}
\end{array}
\end{equation}
Furthermore, we have
\begin{equation}\label{20140228:eq2}
\{{L_0}_\lambda{L_0}\}_{z,\rho}
=
(\partial+2\lambda)L_0
\,.
\end{equation}
In particular, $L_0$ is a Virasoro element of $\mc W$ with zero central charge,
and the generators $w(a),\,a\in\mf g^f_0$, are primary elements with respect to $L_0$
of conformal weight $1$.
\item
We have:
\begin{equation}\label{20140228:eq3}
\{{w(f)}_\lambda{w(f)}\}_{z,\rho}
=
(\partial+2\lambda)w(f)
-(x|x) \lambda^3
+2z(s|f)\lambda
\,,
\end{equation}
i.e. $w(f)\in\mc W\{2\}$ is a Virasoro element, with central charge $-(x|x)$.
Moreover,
\begin{equation}\label{20140228:eq3b}
\{{w(f)}_\lambda L_0\}_{z,\rho}
=
\{{L_0}_\lambda{w(f)}\}_{z,\rho}
=0
\,.
\end{equation}
\item
The element $L=w(f)+L_0\in\mc W\{2\}$ is also a Virasoro element of $\mc W$,
and we have
\begin{equation}\label{20140228:eq4}
\{L_\lambda L\}_{z,\rho}
=
(\partial+2\lambda)L-(x|x)\lambda^3+2z(s|f)\lambda
\,.
\end{equation}
For $a\in\mf g^f_{1-\Delta}$ we have
\begin{equation}\label{20140228:eq5}
\{L_\lambda w(a)\}_{z,\rho}
=
(\partial+\Delta\lambda)w(a)
-\frac{(e|a)}2\lambda^3+z\Delta(s|a)\lambda
\,.
\end{equation}
In particular,
for $z=0$, all the generators $w(a),\,a\in\mf g^f$, of $\mc W$ are primary elements for $L$,
provided that $(e|a)=0$.
In other words, for $z=0$, $\mc W$ is an algebra of differential polynomials
generated by $L$ and $\dim(\mf g^f)-1$ primary elements with respect to $L$.
So, $\mc W$ is a PVA of \emph{CFT type} (cf. \cite{DSKW10}).

\end{enumerate}
\end{proposition}
\begin{proof}
By equation \eqref{20140225:eq1} and the Leibniz rule, we have
\begin{equation}\label{20140228:eq6}
\begin{array}{l}
\displaystyle{
\vphantom{\Big(}
\{{L_0}_\lambda w(a)\}_{z,\rho}
=
\sum_{j\in J^f_0}
{\{{w(q^{j})}_{\lambda+\partial} w(a)\}_{z,\rho}}_\to w(q_{j})
} \\
\displaystyle{
\vphantom{\Big(}
=
\sum_{j\in J^f_0}
w(q_{j})w([q^{j},a])
+\sum_{j\in J^f_0}(q^{j}|a)(\lambda+\partial)w(q_{j})
+z\sum_{j\in J^f_0}(s|[q^{j},a])w(q_{j})
\,.}
\end{array}
\end{equation}
The second term is non-zero only for $a\in\mf g^f_0$, and in this case 
it is $(\lambda+\partial)w(a)$.
The last term is non-zero only for $k=d$, and in this case it is $-zw([s,a]^\sharp)$.
This proves the first equation in \eqref{20140228:eq1}.
The second equation in \eqref{20140228:eq1} is obtained from the first by skewsymmetry.
For $a\in\mf g^f_0$, we have, by a simple symmetry argument, 
that $\sum_{j\in J^f_0}q_{j}[q^{j},a]=0$, as an element of $S^2(\mf g^f_0)$.
Hence, equations \eqref{20140228:eq1b} are a special case of equations \eqref{20140228:eq1}.
Moreover, equation \eqref{20140228:eq2} follows immediately from \eqref{20140228:eq1b}
and the Leibniz rule.
This proves part (a).

Letting $a=f$ in equation \eqref{20140227:eq11} we get
\begin{equation}\label{20140228:eq7}
\begin{array}{l}
\displaystyle{
\phantom{\Big(}
\{{w(f)}_\lambda{w(f)}\}_{z,\rho}
=
\!\!
\sum_{j\in J^f_{\leq-\frac12}}
\!\!
w(q_{j})w([q^{j},f])
+(\partial+2\lambda)w(f)
-(x|x)\lambda^3
+2z(s|f)\lambda
\,.}
\end{array}
\end{equation}
Here we used the facts that $\frac{(e|f)}{2}=(x|x)$ and $[s,f]^\sharp=0$.
To get equation \eqref{20140228:eq3}
we just observe that the first term in the RHS of \eqref{20140228:eq7} is zero.
Indeed, it is easy to check that $\{q_j\}_{j\in J^f_{-\frac12}}$ and $\{[f,q^j]\}_{j\in J^f_{-\frac12}}$
are dual bases of $\mf g^f_{-\frac12}$ with respect to the
non-degenerate skewsymmetric form $(e|[\cdot\,,\,\cdot])$.
But then a simple symmetry argument shows that
$\sum_{j\in J^f_{\leq-\frac12}}q_{j}[q^{j},f]$,
considered as an element of $S^2(\mf g^f_{-\frac12})$, is zero.
To prove equation \eqref{20140228:eq3b} we use \eqref{20140227:eq11}
and the Leibniz rule:
\begin{equation}\label{20140228:eq8}
\begin{array}{l}
\displaystyle{
\vphantom{\Big(}
\{w(f)_\lambda L_0\}_{z,\rho}
=
\sum_{j\in J^f_0}
\{ w(f)_\lambda w(q_j) \}_{z,\rho} w(q^j)
} \\
\displaystyle{
\vphantom{\Big(}
=
\sum_{j\in J^f_0}
\sum_{i\in J^f_{\leq-\frac12}}
w(q_{i})w([q^{i},q_j]^\sharp)
w(q^j)
\,.}
\end{array}
\end{equation}
Equation \eqref{20140228:eq3b} then follows from \eqref{20140228:eq8}
by the observation that, for $j\in J^f_0$ and $i\in J^f_{\leq-\frac12}$,
we have $[q^{i},q_j]\in\mf g_{\geq\frac12}$, so that $[q^{i},q_j]^\sharp=0$.
This proves part (b).

Equation \eqref{20140228:eq4} is an immediate consequence of equations \eqref{20140228:eq2},
\eqref{20140228:eq3} and \eqref{20140228:eq3b}.
Finally, equation \eqref{20140228:eq5}
follows form equations \eqref{20140227:eq11} and \eqref{20140228:eq1}
and the observation that
$\sum_{j\in J^f}q_{j}[q^{j},a]^\sharp$, viewed as an element of $S^2(\mf g^f)$,
is zero.
\end{proof}
Note that the definition of $L$ in Proposition \ref{20140228:prop}(c)
is compatible with the Virasoro element in \cite{DSKV14}.
The fact that $\mc W$ is generated by $L$ and $\dim(\mf g^f)-1$ primary elements
has been known to physicists for a long time \cite{BFOFW90}.
\begin{remark}\label{20140312:rem2}
By equation \eqref{20140228:eq4} the central charge of the Virasoro element $L$ is $c=-(x|x)$,
which varies with the rescaling of the bilinear form $(\cdot\,|\,\cdot)$.
\end{remark}

\section{Isomorphism between the Zhu algebra of
\texorpdfstring{$\mc W(\mf g,f)$}{W(g,f)} and
\texorpdfstring{$\mc W^{\text{fin}}(\mf g,f)$}{Wfin(g,f)}
}
\label{sec:7}

Recall that an \emph{energy operator} $H$
on a Poisson vertex algebra $\mc V$
is a diagonalizable operator on $\mc V$,
which is a derivation of the commutative associative product,
and such that
\begin{equation}\label{20140319:eq1}
H\{a_\lambda b\}
=
\{H(a)_\lambda b\}
+\{a_\lambda H(b)\}
-(1+\lambda\partial_\lambda)\{a_\lambda b\}
\,.
\end{equation}
If $a\in\mc V$ is an eigenvector of $H$, we denote by $\Delta(a)$ the corresponding eigenvalue
(or \emph{conformal weight}).
Given a Poisson vertex algebra $\mc V$ with an energy operator $H$,
following \cite[Sec.6]{DSK06} we introduce
the corresponding $H$-\emph{twisted Zhu algebra} $\Zhu_{z}(\mc V)$.
It is a $1$-parameter family of Poisson algebras (parametrized by $z\in\mb F$)
defined as follows.
As a commutative associative algebra,
\begin{equation}\label{20140319:eq2}
\Zhu{}_{z}(\mc V)
=
\mc V/\langle(\partial+z H)\mc V\rangle_{\mc V}
\,,
\end{equation}
where $\langle(\partial+z H)\mc V\rangle_{\mc V}$ denotes the differential ideal of $\mc V$
generated by the elements $\partial a+z H(a)$, where $a\in\mc V$.
The Poisson bracket on $\Zhu_{z}(\mc V)$ is defined by
\begin{equation}\label{20140319:eq3}
\{a,b\}_{z}
=
\{\tilde{a}\,{}_{z\partial_{\epsilon}}\tilde{b}\}_\to\epsilon^{\Delta(a)-1}\big|_{\epsilon=1}
+\langle(\partial+z H)\mc V\rangle_{\mc V}
\,,
\end{equation}
where 
$\tilde{a},\tilde{b}\in\mc V$ are representatives of $a,b\in\Zhu_{z}(\mc V)$.
Formula \eqref{20140319:eq3} is a special case of \cite[Eq.(6.3)]{DSK06}.

Here we compute the Zhu algebra of the classical affine $\mc W$-algebra $\mc W(\mf g,f)$,
with the energy operator $H$ given by the conformal weight
defined in Section \ref{sec:2.2}:
$H(w(p))=\Delta(p)w(p)$, where $\Delta(p)=(1-\delta(p))$, for $p\in\mf g^f$.
As a commutative associative algebra,
$\Zhu_z(\mc W(\mf g,f))=S(w(\mf g^f))$,
and we have the relation
\begin{equation}\label{20140319:eq4}
\partial A
=
-z\Delta(A)A
\,,
\end{equation}
for every eigenvector $A\in\mc W(\mf g,f)$ of $H$.
\begin{theorem}\label{20140319:thm}
The Poisson bracket on $\Zhu_z(\mc W(\mf g,f))$
is given by the following formula
(for $a\in g^f_{-h}$ and $b\in\mf g^f_{-k}$):
\begin{equation}\label{20140319:eq5}
\begin{array}{l}
\displaystyle{
\phantom{\Big(}
\{w(a), w(b)\}_{z}
=
w([a,b])-z(x|[a,b])
-\sum_{t=1}^\infty
\sum_{
-h+1\leq k_t\prec \dots\prec k_1\leq k
}
\sum_{
(\vec{j},\vec{n})\in J_{-\vec{k}}
}
} \\
\displaystyle{
\phantom{\Big(}
\big(
w([b,q^{j_1}_{n_1}]^\sharp)
-z(x|[b,q^{j_{1}}_{n_{1}}])
\big)
\big(
w([q^{n_1+1}_{j_1},q^{j_2}_{n_2}]^\sharp)
-z(x|[q^{n_1+1}_{j_1},q^{j_2}_{n_2}])
\big)
\dots } \\
\displaystyle{
\phantom{\Big(}
\dots
\big(
w([q^{n_{t-1}+1}_{j_{t-1}},q^{j_t}_{n_t}]^\sharp)
-z(x|[q^{n_{t-1}+1}_{j_{t-1}},q^{j_t}_{n_t}])
\big)
\big(
w([q^{n_t+1}_{j_t},a]^\sharp)
-z(x|[q^{n_t+1}_{j_t},a])
\big)
\,.}
\end{array}
\end{equation}
\end{theorem}
\begin{lemma}\label{20140319:lem}
For $s\geq0$,
$a_1,b_1,\dots,a_s,b_s\in\mf g$ eigenvectors of $\ad x$,
$C\in\mc W(\mf g,f)$ eigenvector of $H$,
and $\alpha\in\frac12\mb Z$,
let
$$
A(\epsilon)
=
\big(
w([a_1,b_1]^\sharp)\!-\!(a_1|b_1)(\partial\!+\!z\partial_\epsilon)
\big)
\dots
\big(
w([a_s,b_s]^\sharp)\!-\!(a_s|b_s)(\partial\!+\!z\partial_\epsilon)
\big)
C\epsilon^\alpha
\,.
$$
Then, modulo the relations \eqref{20140319:eq4}, we have
\begin{equation}\label{20140319:eq8}
(\partial+\!z\partial_\epsilon)A(\epsilon)\big|_{\epsilon=1}
=
z(\alpha-\Delta([a_1,b_1])-\dots-\Delta([a_s,b_s])-\Delta(C))
A(\epsilon)\big|_{\epsilon=1}
\,.
\end{equation}
\end{lemma}
\begin{proof}
For $s=0$ equation \eqref{20140319:eq8} reduces to
\begin{equation}\label{20140319:eq9}
(\partial+z\partial_\epsilon)
C\epsilon^\alpha\big|_{\epsilon=1}
=
z(\alpha-\Delta(C))
C\epsilon^\alpha\big|_{\epsilon=1}
\,,
\end{equation}
which is clear by the relation \eqref{20140319:eq4}.
Next, we prove \eqref{20140319:eq8} for $s=1$. 
We have
\begin{equation}\label{20140319:eq10}
\begin{array}{l}
\displaystyle{
\vphantom{\Big(}
(\partial+z\partial_\epsilon)
\big(
w([a_1,b_1]^\sharp)-(a_1|b_1)(\partial+z\partial_\epsilon)
\big)
C\epsilon^\alpha\big|_{\epsilon=1}
} \\
\displaystyle{
\vphantom{\Big(}
=
\partial w([a_1,b_1]^\sharp)C
- (a_1|b_1) \partial^2 C
- (a_1|b_1) \partial C z \partial_\epsilon  \epsilon^\alpha\big|_{\epsilon=1}
} \\
\displaystyle{
\vphantom{\Big(}
+w([a_1,b_1]^\sharp)C z\partial_\epsilon \epsilon^\alpha\big|_{\epsilon=1}
- (a_1|b_1) \partial C z\partial_\epsilon \epsilon^\alpha\big|_{\epsilon=1}
-(a_1|b_1) z^2 C \partial_\epsilon^2 \epsilon^\alpha\big|_{\epsilon=1}
} \\
\displaystyle{
\vphantom{\Big(}
=
z(\alpha-\Delta([a_1,b_1])-\Delta(C)) w([a_1,b_1]^\sharp)C
} \\
\displaystyle{
\vphantom{\Big(}
-z^2
\big(
\alpha(\alpha-1)-2\alpha\Delta(C)+\Delta(C)(\Delta(C)+1)
\big)
(a_1|b_1) C
\,.}
\end{array}
\end{equation}
On the other hand, we have
\begin{equation}\label{20140319:eq11}
\begin{array}{l}
\displaystyle{
\vphantom{\Big(}
z(\alpha-\Delta([a_1,b_1])-\Delta(C))
\big(
w([a_1,b_1]^\sharp)-(a_1|b_1)(\partial+z\partial_\epsilon)
\big)
C\epsilon^\alpha\big|_{\epsilon=1}
} \\
\displaystyle{
\vphantom{\Big(}
=
z(\alpha-\Delta([a_1,b_1])-\Delta(C))
w([a_1,b_1]^\sharp)C
} \\
\displaystyle{
\vphantom{\Big(}
-
z^2(\alpha-\Delta([a_1,b_1])-\Delta(C))(\alpha-\Delta(C))
(a_1|b_1) C
\,.}
\end{array}
\end{equation}
Note that, if $(a_1|b_1)\neq0$, then $\Delta([a_1,b_1])=1$.
Hence, comparing equations \eqref{20140319:eq10} and \eqref{20140319:eq11},
we get that
\begin{equation}\label{20140319:eq12}
\begin{array}{l}
\displaystyle{
\vphantom{\Big(}
(\partial+z\partial_\epsilon)
\big(
w([a_1,b_1]^\sharp)-(a_1|b_1)(\partial+z\partial_\epsilon)
\big)
C\epsilon^\alpha\big|_{\epsilon=1}
} \\
\displaystyle{
\vphantom{\Big(}
=
z(\alpha-\Delta([a_1,b_1])-\Delta(C))
\big(
w([a_1,b_1]^\sharp)-(a_1|b_1)(\partial+z\partial_\epsilon)
\big)
C\epsilon^\alpha\big|_{\epsilon=1}
\,,}
\end{array}
\end{equation}
which is the same as \eqref{20140319:eq8} with $s=1$.
The general formula \eqref{20140319:eq8} for arbitrary $s\geq1$
follows by equations \eqref{20140319:eq9} and \eqref{20140319:eq12}
and an easy induction.
\end{proof}
\begin{proof}[{Proof of Theorem \ref{20140319:thm}}]
According to equation \eqref{20140304:eq4} and formula \eqref{20140319:eq3},
for $a\in\mf g^f_{-h}$ and $b\in\mf g^f_{-k}$
the Poisson bracket $\{w(a),w(b)\}_z$ is given by:
\begin{equation}\label{20140304:eq4-z}
\begin{array}{l}
\displaystyle{
\phantom{\Big(}
w([a,b])+z(a|b)\partial_\epsilon\epsilon^{h}\Big|_{\epsilon=1}
-\sum_{t=1}^\infty
\sum_{
-h+1\leq k_t\prec \dots\prec k_1\leq k
}
\sum_{
(\vec{j},\vec{n})\in J_{-\vec{k}}
}
} \\
\displaystyle{
\phantom{\Big(}
\big(
w([b,q^{j_1}_{n_1}]^\sharp)
-(b|q^{j_1}_{n_1})(z\partial_\epsilon+\partial)
\big)
\big(
w([q^{n_1+1}_{j_1},q^{j_2}_{n_2}]^\sharp)
-(q^{n_1+1}_{j_1}|q^{j_2}_{n_2}) (z\partial_\epsilon+\partial)
\big)
\dots } \\
\displaystyle{
\phantom{\Big(}
\dots
\big(
w([q^{n_t+1}_{j_t},a]^\sharp)
-(q^{n_t+1}_{j_t}|a) z\partial_\epsilon
\big)
\epsilon^{h}\Big|_{\epsilon=1}
\,,}
\end{array}
\end{equation}
modulo the relations \eqref{20140319:eq4}.
We clearly have
$$
z(a|b)\partial_\epsilon\epsilon^{h}\Big|_{\epsilon=1}
=
z(a|b)h
=
-z(x|[a,b])
\,.
$$
By Lemma \ref{20140319:lem},
in the first factor of \eqref{20140304:eq4-z} we can replace
$(b|q^{j_1}_{n_1})(z\partial_\epsilon+\partial)$
by
\begin{equation}\label{20140319:eq13}
(b|q^{j_1}_{n_1})
\big(
h
-\Delta([q^{n_1+1}_{j_1},q^{j_2}_{n_2}])
-\dots
-\Delta([q^{n_{t-1}+1}_{j_{t-1}},q^{j_t}_{n_t}])
-\Delta([q^{n_t+1}_{j_t},a])
\big)
\,.
\end{equation}
But, for $(j_j,n_i)\in J_{-k_i}$, we have
$$
\Delta([q^{n_i+1}_{j_i},q^{j_{i+1}}_{n_{i+1}}])
=
k_i-k_{i+1}
\,.
$$
Hence \eqref{20140319:eq13} becomes
$$
(b|q^{j_1}_{n_1})
\big(
h
-(k_1-k_2)
-\dots
-(k_{t-1}-k_t)
-(h+k_t)
\big)
=
-k_1(b|q^{j_1}_{n_1})
=
(x|[b,q^{j_1}_{n_1}])
\,.
$$
Similarly for all the other factors.
Equation \eqref{20140319:eq5} follows.
\end{proof}
For $z=0$ 
the Zhu algebra reduces to 
$\Zhu_{z=0}\mc W(\mf g,f)=\mc W(\mf g,f)/\langle\partial\mc W(\mf g,f)\rangle$,
and formula \eqref{20140319:eq5} reduces to formula \eqref{20140319:eq7},
once we identify $p\in\mf g^f$ with $w(p)\in\mc W(\mf g,f)$.
Hence the Poisson algebras $\Zhu_{z=0}\mc W(\mf g,f)$ and $\mc W^{\text{fin}}(\mf g,f)$ 
are isomorphic.
More generally,
it is immediate to see that formula \eqref{20140319:eq5}
is the same as formula \eqref{20140319:eq6}
with $z$ replaced by $-z$.
Hence, as a consequence of Theorem \ref{20140318:rem}
and Corollary \ref{20140318:remb},
we get the following
\begin{corollary}\label{20140319:cor}
The Poisson algebras 
$\Zhu_z(\mc W(\mf g,f))$ and $\mc W^{\text{fin}}(\mf g,f)$
are isomorphic for every value of $z\in\mb F$.
In fact, the Poisson bracket \eqref{20140319:eq5}
is unchanged if we replace $-zx$ by $\frac{z^2}{4}e$,
and we thus have an explicit isomorphism 
$\Zhu_{z=0}(\mc W(\mf g,f))\to\Zhu_z(\mc W(\mf g,f))$
given by $w(q)\mapsto w(q)+\frac{z^2}{4}(q|e)$, for $q\in\mf g^f$.
\end{corollary}
\begin{remark}\label{20140319:rem2}
The ``quantum'' version of Corollary \ref{20140319:cor}
was established in \cite{DSK06}:
$\Zhu_{z}W(\mf g,f)\simeq W^{\text{fin}}(\mf g,f)$ for $z\neq0$.
(As before, $W$, as opposed to $\mc W$, refers to ``quantum'' $W$-algebras.)
Taking $z=0$, Corollary \eqref{20140319:cor} shows, in particular, that
the Poisson algebras $\mc W^{\text{fin}}(\mf g,f)$,
$\mc W(\mf g,f)/\langle\partial\mc W(\mf g,f)\rangle$,
and $W(\mf g,f)/\langle\partial W(\mf g,f)\rangle$,
are all isomorphic
(the last isomorphism is proved in \cite[Sec.6]{DSK06}),
and the quantum finite $W$-algebra $W^{\text{fin}}(\mf g,f)$ is their quantization.
Note that in \cite{DSK06} we use the cohomological definition of
classical and quantum $W$-algebras.
The equivalence of these definitions to the Hamiltonian reduction definitions
was established in the appendix of \cite{DSK06} for the finite quantum $W$-algebra,
and in \cite{Suh13} for the classical ones.
\end{remark}

By Corollary \ref{20140319:cor},
we can view the classical finite $\mc W$-algebra $\mc W^{\text{fin}}(\mf g,f)$
as the Zhu-algebra of the classical $\mc W$-algebra $\mc W(\mf g,f)$ at $z=0$.
It follows that
$\mc W^{\text{fin}}(\mf g,f)=\mc W/\mc W\partial\mc W$,
can be obtained by classical Hamiltonian reduction as in the affine case:
$$
\mc W^{\text{fin}}(\mf g,f)
=\big\{g\in S(\mf g_{\leq\frac12})\,\big|\,\rho((\ad a)(g))=0\,\text{ for all }a\in\mf g_{\geq\frac12}\}
\,,
$$
where $\rho:\, S(\mf g)\twoheadrightarrow S(\mf g_{\leq\frac12})$
is the algebra homomorphism defined on generators by \eqref{rho}.
The analogue of Corollary \ref{20140221:cor} holds in this case as well:
\begin{corollary}\label{20140311:cor}
For every $q\in\mf g^f$ there exists a unique
element $w=w(q)\in\mc W^{\text{fin}}(\mf g,f)$ 
of the form $w=q+r$, where 
$r$ lies in the ideal of $S(\mf g_{\leq\frac12})$ generated by $[e,\mf g_{\leq-\frac12}]$,
and it is homogeneous with respect to conformal weight
provided that $q$ is an $\ad x$-eigenvector.
Consequently, $\mc W^{\text{fin}}(\mf g,f)$ coincides with the algebra of differential polynomials in the variables
$w(q_j)$, where $\{q_j\}$ is a basis of $\mf g^f$.
\end{corollary}
\begin{remark}\label{20140311:rem}
The canonical quotient map $S(\mf g)\to S(\mf g)/\langle m-(f|m)\,|\,m\in\mf g_{\geq\frac12}\rangle$
induces, for a principal nilpotent element $f_{pr}\in\mf g$,
an isomorphism \cite{Kos78}
$$
\phi:\,S(\mf g)^{\mf g}
\stackrel{\sim}{\longrightarrow}
\big(S(\mf g)/\langle m-(f_{pr}|m)\,|\,m\in\mf g_{\geq\frac12}\rangle\big)^{\ad\mf n}=\mc W^{\text{fin}}(\mf g,f_{pr})
\,.
$$
Recall that $\mf g^{f_{pr}}$ has a basis $\{q_j\}_{j=1}^\ell$ consisting of $\ad x$-eigenvectors
with eigenvalues $m_1=1<m_2<\dots<m_\ell$, where the $m_i$'s are the exponents of $\mf g$.
(Only in the case of $\mf g$ of type $D_{2n}$ two of the exponents are equal, both being $n$.)
Hence, $\{\phi^{-1}(w(q_j))\}_{j=1}^\ell$ form a canonical (up to a scalar factor for each basis element)
set of generators of the algebra $S(\mf g)^{\mf g}$
(with the mentioned above exception of $D_{2n}$).
\end{remark}

\section{The generalized Miura map for classical 
\texorpdfstring{$\mc W$}{W}-algebras}
\label{sec:daniele}

Consider the affine PVA $\mc V(\mf g)$ with $\lambda$-bracket 
\eqref{lambda} with $z=0$.
We denote by $\{\cdot\,_{\lambda}\,\cdot\}^0$ the restriction of this $\lambda$-bracket
to the PVA subalgebra $\mc V(\mf g_{\leq0})=S(\mb F[\partial]\mf g_{\leq0})$.
Furthermore, let $\mc F(\mf g_{\frac12})$ be 
the algebra of differential polynomials $S(\mb F[\partial]\mf g_{\frac12})$,
endowed with the PVA $\lambda$-bracket defined, on generators, by:
\begin{equation}\label{20140318_dan:eq1}
\{a_{\lambda}b\}^{ne}=-(f|[a,b])=:\langle a|b\rangle
\,,
\qquad\text{for all }a,b\in\mf g_{\frac12}\,.
\end{equation}
We then consider the tensor product of PVA's
$$
\mc V=\mc V(\mf g_{\leq0})\otimes\mc F(\mf g_{\frac12})
\,.
$$
Namely, the $\lambda$-brackets on generators are defined by
$$
\begin{array}{l}
\displaystyle{
\vphantom{\Big(}
\{a_\lambda b\}^\otimes=\{a_\lambda b\}^0=[a,b]+(a|b)\lambda
\,\,\text{ for }\,\,
a,b\in\mf g_{\leq0}
\,,} \\
\displaystyle{
\vphantom{\Big(}
\{a_\lambda c\}^\otimes=\{c_\lambda a\}^\otimes=0
\,\,\text{ for }\,\,
a\in\mf g_{\leq0},\,c\in\mf g_{\frac12}
\,,} \\
\displaystyle{
\vphantom{\Big(}
\{c_\lambda d\}^\otimes=\{c_\lambda d\}^{ne}=-(f|[c,d])
\,\,\text{ for }\,\,
c,d\in\mf g_{\frac12}
\,.}
\end{array}
$$
\begin{theorem}\label{20140318:prop1}
The obvious differential algebra isomorphism
$\mc V(\mf g_{\leq\frac12})\stackrel{\sim}{\longrightarrow}\mc V$
restricts to an (injective) PVA homomorphism 
\begin{equation}\label{20140320:eq1}
\mc W=\mc W(\mf g,f)
\hookrightarrow
\mc V=\mc V(\mf g_{\leq0})\otimes\mc F(\mf g_{\frac12})
\,.
\end{equation}
\end{theorem}
\begin{proof}
Recall that $\mc W=\mc W(\mf g,f)$ is a differential subalgebra of $\mc V(\mf g_{\leq\frac12})$,
and we have an obvious isomorphism $\mc V\simeq\mc V(\mf g_{\leq\frac12})$.
Hence, we have an injective differential algebra homomorphism 
$\mc W\hookrightarrow\mc V$.
We need to show that
\begin{equation}\label{20140318_dan:eq2}
\{g_\lambda h\}_\rho=\{g_\lambda h\}^{\otimes}\,,
\end{equation}
for every $g,h\in\mc W$.

Let $\{q_j\}_{j\in J_{\leq\frac12}}$ be a basis of $\mf g_{\leq\frac12}$ such
that $\{q_j\}_{j\in J_{\leq0}}$ is a basis of $\mf g_{\leq0}$
and $\{q_j\}_{j\in J_{\frac12}}$ is a basis of $\mf g_{\frac12}$
(hence, $J_{\leq\frac12}=J_{\leq0}\cup J_{\frac12}$).
Recall that, due to sesquilinearity and Leibniz rule, we have the Master Formula for $\lambda$-brackets
of arbitrary differential polynomials (see \cite[Ex.6.2]{DSK06}),
which we use below.
By the definition of the tensor $\lambda$-bracket in $\mc V$, 
the RHS of \eqref{20140318_dan:eq2} becomes
\begin{align}
\begin{split}\label{20140319_dan:eq1}
\{g_\lambda h\}^\otimes
=\sum_{\substack{i,j\in J_{\leq0}\\m,n\in\mb Z_+}}
\frac{\partial h}{\partial q_j^{(n)}}(\lambda+\partial)^n
\{q_i{}_{\lambda+\partial}q_j\}^0
(-\lambda-\partial)^m\frac{\partial g}{\partial q_i^{(m)}}\\
+\sum_{\substack{i,j\in J_{\frac12}\\m,n\in\mb Z_+}}
\frac{\partial h}{\partial q_j^{(n)}}(\lambda+\partial)^n
\{q_i{}_{\lambda+\partial}q_j\}^{ne}
(-\lambda-\partial)^m\frac{\partial g}{\partial q_i^{(m)}}\,.
\end{split}
\end{align}
On the other hand,
by equation \eqref{20120511:eq3} and the definition
of the map $\rho$ given by equation \eqref{rho},
the LHS\eqref{20140318_dan:eq2} becomes
\begin{align}
\begin{split}\label{20140318_dan:eq3}
\{g_\lambda h\}_\rho
=\sum_{\substack{i,j\in J_{\leq0}\\m,n\in\mb Z_+}}
\frac{\partial h}{\partial q_j^{(n)}}(\lambda+\partial)^n
\rho\{q_i{}_{\lambda+\partial}q_j\}
(-\lambda-\partial)^m\frac{\partial g}{\partial q_i^{(m)}}\\
+\sum_{\substack{i\in J_{\leq0},j\in J_{\frac12}\\m,n\in\mb Z_+}}
\frac{\partial h}{\partial q_j^{(n)}}(\lambda+\partial)^n
\rho\{q_i{}_{\lambda+\partial}q_j\}
(-\lambda-\partial)^m\frac{\partial g}{\partial q_i^{(m)}}\\
+\sum_{\substack{i\in J_{\frac12},j\in J_{\leq0}\\m,n\in\mb Z_+}}
\frac{\partial h}{\partial q_j^{(n)}}(\lambda+\partial)^n
\rho\{q_i{}_{\lambda+\partial}q_j\}
(-\lambda-\partial)^m\frac{\partial g}{\partial q_i^{(m)}}\\
+\sum_{\substack{i,j\in J_{\frac12}\\m,n\in\mb Z_+}}
\frac{\partial h}{\partial q_j^{(n)}}(\lambda+\partial)^n
\rho\{q_i{}_{\lambda+\partial}q_j\}
(-\lambda-\partial)^m\frac{\partial g}{\partial q_i^{(m)}}\,.
\end{split}
\end{align}
We note that for $i,j\in J_{\leq0}$ we have
$\rho\{q_i{}_\lambda q_j\}=\{q_i{}_\lambda q_j\}^0$. Hence,
the first summand in \eqref{20140318_dan:eq3} becomes
\begin{equation}\label{20140318:eq4}
\sum_{\substack{i,j\in J_{\leq0}\\m,n\in\mb Z_+}}
\frac{\partial h}{\partial q_j^{(n)}}(\lambda+\partial)^n
\{q_i{}_{\lambda+\partial}q_j\}^0
(-\lambda-\partial)^m\frac{\partial g}{\partial q_i^{(m)}}\,.
\end{equation}
Furthermore, by definition of $\mc W$, we have
$\rho\{g_{\lambda} q_j\}=0$,
for every $j\in J_{\frac12}$.
Hence, using the fact that
$J_{\leq\frac12}=J_{\leq0}\cup J_{\frac12}$ we get
\begin{align}
\begin{split}
\label{20140319_dan:eq2}
0=\rho\{g_{\lambda} q_j\}
=\sum_{\substack{i\in J_{\leq0}\\m\in\mb Z_+}}\rho\{q_i{}_{\lambda+\partial}q_j\}
(-\lambda-\partial)^m\frac{\partial g}{\partial q_{i}^{(m)}}\\
+\sum_{\substack{i\in J_{\frac12}\\m\in\mb Z_+}}\rho\{q_i{}_{\lambda+\partial}q_j\}
(-\lambda-\partial)^m\frac{\partial g}{\partial q_{i}^{(m)}}\,.
\end{split}
\end{align}
Using the identity \eqref{20140319_dan:eq2}, the second summand in equation
\eqref{20140318_dan:eq3} becomes
\begin{equation}\label{20140319_dan:eq3}
-\sum_{\substack{i,j\in J_{\frac12}\\m,n\in\mb Z_+}}
\frac{\partial h}{\partial q_j^{(n)}}(\lambda+\partial)^n
\rho\{q_i{}_{\lambda+\partial}q_j\}
(-\lambda-\partial)^m\frac{\partial g}{\partial q_i^{(m)}}\,.
\end{equation}
Similarly, $\rho\{q_i{}_\lambda h\}=0$, for every $i\in J_{\frac12}$, from which
follows that
\begin{equation}\label{20140319_dan:eq4}
\sum_{\substack{j\in J_{\leq0}\\n\in\mb Z_+}}
\frac{\partial h}{\partial q_j^{(n)}}(\lambda+\partial)^n
\rho\{q_i{}_{\lambda} q_j\}
=-
\sum_{\substack{j\in J_{\frac12}\\n\in\mb Z_+}}
\frac{\partial h}{\partial q_j^{(n)}}(\lambda+\partial)^n
\rho\{q_i{}_{\lambda}q_j\}
\,.
\end{equation}
Using equation \eqref{20140319_dan:eq4} (where we replace
$\lambda$ with $\lambda+\partial$
acting on the right), we get that the contribution of the third summand
in equation \eqref{20140318_dan:eq3} is
\begin{equation}\label{20140318:eq6}
-\sum_{\substack{i,j\in J_{\frac12}\\m,n\in\mb Z_+}}
\frac{\partial h}{\partial q_j^{(n)}}(\lambda+\partial)^n
\rho\{q_i{}_{\lambda+\partial}q_j\}
(-\lambda-\partial)^m\frac{\partial g}{\partial q_i^{(m)}}\,.
\end{equation}
Combining equations \eqref{20140318:eq4}, \eqref{20140319_dan:eq3} and
\eqref{20140318:eq6}, it follows that the equation \eqref{20140318_dan:eq3}
becomes
\begin{align*}
\{g_\lambda h\}_\rho
=\sum_{\substack{i,j\in J_{\leq0}\\m,n\in\mb Z_+}}
\frac{\partial h}{\partial q_j^{(n)}}(\lambda+\partial)^n
\{q_i{}_{\lambda+\partial}q_j\}^\otimes
(-\lambda-\partial)^m\frac{\partial g}{\partial q_i^{(m)}}\\
-\sum_{\substack{i,j\in J_{\frac12}\\m,n\in\mb Z_+}}
\frac{\partial h}{\partial q_j^{(n)}}(\lambda+\partial)^n
\rho\{q_i{}_{\lambda+\partial}q_j\}
(-\lambda-\partial)^m\frac{\partial g}{\partial q_i^{(m)}}\,.
\end{align*}
To prove that the above expression is the same as equation
\eqref{20140319_dan:eq1} it suffices to note that
$\rho\{q_i{}_\lambda q_j\}=(f|[q_i,q_j])=-\{q_i{}_\lambda q_j\}^{ne}$,
for $i,j\in\ J_{\frac12}$.
\end{proof}
\begin{corollary}
The homomorphism \eqref{20140320:eq1}
induces an injective PVA homomorphism 
$\mu:\,\mc W\to\mc V(\mf g_0)\otimes\mc F(\mf g_{\frac12})$,
called the \emph{generalized Miura map}.
\end{corollary}
\begin{proof}
Composing the PVA homomorphism \eqref{20140320:eq1}
with the projection $\mc V(\mf g_{\leq0})\to\mc V(\mf g_0)$ (which is also a PVA homomorphism),
we get a PVA homomorphism $\mc W\to\mc V(\mf g_0)\otimes\mc F(\mf g_{\frac12})$.
It is not difficult to show, using Theorem \ref{20140221:thm1},
that, for $j\in J^f_{-k}$,
the term of $w(q_j)$ in $\mf g_0\oplus\mf g_{\frac12}$ 
is equal to $(-\partial)^nq_j^n\neq0$, where $n=[k+\frac12]$.
Injectiveness follows.
\end{proof}
\begin{example}
The Virasoro element $L\in\mc W(\mf g,f)$ from Proposition \ref{20140228:prop}(c)
has the following explicit expression, as an element of $\mc V(\mf g_{\leq\frac12})$ 
(see \cite[Eq.2.19]{DSKV14}):
$$
L=f+x'+\frac12\sum_{i}a^ia_i
+\sum_{k}v^k[f,v_k]
+\frac12\sum_{k}v^k\partial v_k
\,,
$$
where $\{a_i\}$ and $\{a^i\}$ are dual (w.r.t. $(\cdot\,|\,\cdot)$) bases of $\mf g_0$,
and $\{v_k\}$ and $\{v^k\}$ are bases of $\mf g_{\frac12}$
dual with respect to $\langle\cdot\,|\,\cdot\rangle$ (cf. \eqref{20140318_dan:eq1}),
in the sense that $\langle v^h|v_k\rangle=\delta_{h,k}$.
Applying to $L$ the map $\mu:\,\mc W\to\mc V(\mf g_0)\otimes\mc F(\mf g_{\frac12})$,
we thus get the element (cf. \cite[Thm.5.2]{KW04})
$$
\mu(L)
=
x'+\frac12\sum_{i\in J_0}a^ia_i
+\frac12\sum_{k\in J_{\frac12}}v^k\partial v_k
\,.
$$
In the special case of $\mf g=\mf{sl}_2$, we get $\mu(L)= x'+\frac{x^2}{2(x|x)}$,
which is the classical Miura map.
\end{example}

\begin{remark}\label{20140402:rem}
As pointed out in the introduction,
the assumption that $\mf g$ is a simple Lie algebra is not essential.
In fact, all the results of the present paper hold for an arbitrary finite-dimensional Lie algebra 
(or superalgebra) $\mf g$, endowed with a non-degenerate symmetric invariant
bilinear form $(\cdot\,|\,\cdot)$, and an $\mf{sl}_2$ triple $\mf s\subset\mf g$.
\end{remark}

\bigskip
\footnotesize
\noindent\textit{Acknowledgments.}
We wish to thank Pavel Etingof for enlightening discussions 
and Ivan Losev for correspondence.
We are also grateful to IHES, University of Rome and MIT for their kind hospitality.
The first author is supported by the national FIRB grant RBFR12RA9W ``Perspectives in Lie Theory''. 
The third author is supported
by the ERC grant ``FroM-PDE: Frobenius Manifolds and Hamiltonian Partial Differential Equations''.


%
\end{document}